\documentclass[journal]{IEEEtran} 
\IEEEoverridecommandlockouts
\usepackage{etoolbox}

\usepackage[ruled,linesnumbered]{algorithm2e}
\usepackage{tabularx}
\usepackage{booktabs} 
\usepackage{longtable}
\usepackage{multicol}
\usepackage{multirow} 
\usepackage{graphicx}  
\usepackage{float}  
\usepackage{subcaption}

\usepackage{amsmath,amssymb,amsthm}  

\usepackage{stmaryrd}
\usepackage{bm}
\usepackage{array}
\usepackage{xcolor}

\usepackage{makecell}
\usepackage{diagbox}
\usepackage{placeins}
\usepackage{cite} 

\newtheorem{definition}{Definition}
\newtheorem{theorem}{Theorem}
\newtheorem{lemma}{Lemma}

\newtheorem{proposition}{Proposition}

\newtheorem{property}{Property}

\AtBeginEnvironment{definition}{\pushQED{\qed}}
\AtEndEnvironment{definition}{\popQED}

\usepackage{tikz}
\usetikzlibrary{automata, positioning, calc, arrows.meta, shapes,fit, backgrounds}
\usepackage[colorlinks=true, allcolors=blue]{hyperref}

\title{Evolution-Based Timed Opacity under a Universal Observation Model}
\author{Zhe Zhang, Martijn Goorden, and Michel Reniers%
\thanks{The authors are with the Department of Mechanical Engineering, Eindhoven University of Technology, Eindhoven, The Netherlands.
Email: {\tt\small \{z.zhang4, m.a.goorden, m.a.reniers\}@tue.nl}}
\thanks{This work is the full version of a paper submitted to the 65th IEEE Conference on Decision and Control (CDC 2026), providing detailed technical proofs omitted from the conference version due to space constraints.}}

\begin{document}  
\maketitle
\thispagestyle{empty}
\pagestyle{empty}
\begin{abstract}
Existing literature on timed opacity uses specific definitions for restricted subclasses of timed automata or limited observation models. This lack of a unified definition makes it difficult to establish formal relationships and compare the expressiveness of different opacity variants. This paper establishes a unified framework for timed opacity by introducing a universal observation model for timed automata. First, we introduce an observation model with full observation of time delay and partial observation of locations, clocks, and events. Second, based on this model, we define the notion of evolution-based timed opacity. Third, we mathematically prove that evolution-based timed opacity strictly implies language-based timed opacity and establish a formal equivalence with execution-time opacity under constrained observations. This framework establishes a unified semantic hierarchy for characterizing the landscape of timed opacity.
\end{abstract}

\section{Introduction}
In cyber-physical systems (CPS), information security is critical, particularly for safety-critical infrastructures operating under strict timing constraints~\cite{HARKAT2024109891}. Beyond traditional access control, \emph{opacity} provides a rigorous confidentiality guarantee, ensuring that the ``secret'' behaviors of a system remain plausibly deniable to external observers~\cite{Saboori2012}. While opacity has been extensively studied in discrete-event systems (DES)~\cite{Lin2011, Wintenberg2022,Xie2022}, the integration of dense time introduces significant complexity~\cite{cassez2009dark}. In real-time systems modeled by \emph{timed automata} (TA)~\cite{alur1994theory}, the precise timing of events serves as a critical information channel, potentially leaking secrets that would remain hidden in an untimed abstraction.

Since Cassez established the ``dark side'' of timed opacity, specifically that the problem is generally undecidable for TA~\cite{cassez2009dark}, subsequent research has largely focused on identifying a ``bright side'' where verification becomes decidable~\cite{andre2024bright}. To achieve decidability, these efforts typically adopt a pragmatic ``bottom-up'' strategy: they define opacity tailored to specific subclasses of TA or restricted observation models~\cite{Klein2024}. On the system modeling side, structural constraints are often utilized to ensure decidability. For instance,~\cite{An2025} explored structural limits, defining opacity within specific subclasses like event-recording TA, while~\cite{deng2025initial} achieved decidability by focusing on integer reset TA. Similarly, other works simplify the timing structure by associating constraints directly with transitions, such as the real-time labeled TA used by~\cite{Zhang2024} and the time interval automata used by~\cite{Marques2023}. These structures implicitly reset the global clock at each step, effectively limiting the temporal memory to the immediate transition. On the observation side, complexity is often managed by restricting the observer's capabilities. For instance,~\cite{Andr2023} assumes a constrained observer model where observability is limited to execution times. The work in~\cite{Ammar2021} further restricts the observer’s power by imposing a bounded time horizon, where safety is guaranteed within a specific finite duration. Furthermore, approaches in~\cite{Klein2024} simplify the observation model to discrete-time semantics, effectively abstracting away the dense nature of time to facilitate verification. While successful in their specific domains, this ``decidability-first'' approach has led to diverse definitions designed for specific contexts. Existing definitions largely generalize the concepts of untimed opacity, extending confidentiality to the timed domain primarily through the use of timed languages (e.g., language-based timed opacity).

However, this focus on specific subclasses has resulted in a lack of theoretical cohesion. Existing definitions, including language-based timed opacity and execution-time opacity, are often formulated independently to suit specific constraints, resulting in a fragmented landscape. This fragmentation presents two structural challenges: (1) there is a lack of a systematic understanding of how these diverse opacity notions relate to one another; and (2) current formalisms suffer from limited expressive power, failing to provide a unified basis for reasoning about opacity when intruders possess diverse observational capabilities. For instance, existing language-based abstractions exhibit a structural limitation, which we term \emph{suffix blindness}. Because the abstraction of timed words discards any time progress after the last observable discrete event, it fails to capture confidential behaviors leaked through trailing delays. To address these challenges, we propose a top-down semantic framework. Rather than restricting the model to ensure decidability, we establish a comprehensive semantic baseline that inherently captures the continuous dynamics of TA. The main contributions of this paper are threefold:
\begin{itemize}
  \item We introduce an observation model for TA that accounts for the partial observability of locations, clocks, and events (Section~\ref{sec:observation model}). 
  \item Based on the observation model, we formulate a new opacity notion, evolution-based timed opacity (EBTO), that rigorously captures continuous state evolution and dense time dynamics (Section~\ref{sec:notion of EBTO}). 
  \item Finally, we establish a formal hierarchy of timed opacity. We prove that EBTO provides a strictly stronger confidentiality guarantee than language-based timed opacity (LBTO), and demonstrates structural equivalence with execution-time opacity (ETO) under constrained observations (Section~\ref{sec:translation}). 
\end{itemize}
This systematic perspective allows existing opacity problems to be analyzed not as isolated cases, but as specific applications of a unified theory.



\section{Preliminaries}\label{sec:preliminaries}
A TA is a finite automaton extended with a finite set of real-valued clocks. The temporal behaviors of a TA are captured by clock constraints, which serve as guards on edges and invariants on locations. To define TA formally, we first establish the syntax of these constraints.

Given a finite set of real-valued clocks $C$, a \emph{clock constraint}~\cite{alur1994theory} over $C$ is defined as:
\[
\varphi ::= \text{true} \mid x \sim n \mid x - y \sim n \mid \varphi \land \varphi \mid \varphi \lor \varphi
\]
where $x, y \in C$, $n \in \mathbb{N}$, and the comparison operator $\sim$ belongs to the set $\{<, \leq, =, \geq, >\}$. We denote the set of all such clock constraints over $C$ by $\mathcal{G}(C)$ .

While constraints define the logical conditions for timing, the state of the system's time is captured by clock valuations. The following definition formalizes how clocks hold values and how these values change through resets and time elapsing.

Given a finite set of clocks $C$, a \emph{clock valuation} is a function $u : C \rightarrow \mathbb{R}_{\geq 0}$ that assigns a non-negative real value to each clock $x \in C$. We denote by $\mathcal{U}$ the set of all such valuations, i.e., $\mathcal{U} = C \to \mathbb{R}_{\ge 0}$. We denote by $\mathbf{0}$ the \emph{zero valuation}, i.e., the unique valuation in $\mathcal{U}$ such that $\mathbf{0}(x) = 0$ for all $x \in C$.

We define the following operations on clock valuations:
\begin{enumerate}
    \item Reset: Given a set of clocks $r \subseteq C$, the \emph{reset} of $u$ on $r$, denoted $u[r]$, is the valuation defined by:
    \[
    u[r](x) = 
    \begin{cases}
        0 & \text{if } x \in r, \\
        u(x) & \text{otherwise}.
    \end{cases}
    \]
    Intuitively, clocks in $r$ are reset to zero, while others remain unchanged.

    \item Time elapse: Given a delay $\tau \in \mathbb{R}_{\geq 0}$, the \emph{time-elapsed} valuation $u + \tau$ is defined as:
    \[
    (u + \tau)(x) = u(x) + \tau \quad \text{for all } x \in C.
    \]
    This operation models the uniform elapsing of time for all clocks.
\end{enumerate}

We say a clock valuation $u$ satisfies a clock constraint $\varphi$, denoted as $u\models\varphi$, if $\varphi$ is \emph{true} for the values assigned by $u$ to each clock. With these preliminaries established, we can now formally define the structure of a TA.

A \emph{TA}~\cite{alur1994theory} is a six tuple $(L, L_0, C, \Sigma, E, I)$ where  
  \begin{enumerate}
    \item $L$ is a finite set of locations;
   \item $L_0 \subseteq L$ is a finite set of initial locations;
    \item $C$ is a finite set of clocks;
    \item $\Sigma$ is a finite set of events;
    \item $E \subseteq L \times \Sigma \times \mathcal{G}(C) \times 2^{C} \times L$ is a finite set of edges;
    \item $I : L \rightarrow \mathcal{G}(C)$ assigns to each location $l \in L$ an \emph{invariant} $I(l)$, which is a clock constraint from $\mathcal{G}(C)$.
  \end{enumerate}

To analyze the internal operational semantics of a TA, specifically how it transitions between states via discrete events or time delays, we introduce the \emph{semantic graph} (SG) of a TA. The SG represents the infinite state space of a TA, where states consist of locations coupled with concrete clock valuations.

\begin{definition}[SG~\cite{Rashidinejad2024}]\label{SG}
Given a TA $\mathcal{A} = (L, L_0, C, \Sigma, E, I)$, its \emph{SG} $\mathcal{T} = (S, S_0, \Gamma, \delta)$ denoted as $\mathcal{S}(\mathcal{A})$:
\begin{enumerate}
     \item $S = \{(l, u) \in L \times \mathcal{U} \mid u \models I(l)\}$ is the set of states;
     \item $S_0 = \{(l_0, \mathbf{0}) \mid l_0 \in L_0 \land \mathbf{0} \models I(l_0)\}$ is the set of initial states. If $S_0 = \varnothing$, the semantic graph is undefined;
    \item $\Gamma = \Sigma \cup \mathbb{R}_{\geq 0}$ is the a set of actions;
    \item $\delta \subseteq S \times \Gamma \times S$ is the transition relation, defined as the smallest relation satisfying:
    \begin{itemize}
        \item (\emph{Event transition})  
        $((l_s, u_s), \sigma, (l_t, u_s[r])) \in \delta$ if there exists $(l_s, \sigma, g, r, l_t) \in E$, such that $u_s \models g\wedge I(l_s)$, and $u_s[r] \models I(l_t)$;
        
        \item (\emph{Time transition})  
        $((l, u), \tau, (l, u + \tau)) \in \delta$ if $\forall \theta \in [0, \tau]$, $u + \theta \models I(l)$.\qedhere
    \end{itemize}
\end{enumerate}
\end{definition}

A fundamental feature of continuous time semantics is that the passage of time can be arbitrarily split or combined, and zero-time delays do not alter the system state. These properties are known as time additivity, continuity, and reflexivity~\cite{fares2013automatic}.

\begin{property}[Time additivity, continuity, and reflexivity]\label{property:time additivity}
Let $\mathcal{A} = (L, L_0, C, \Sigma, E, I)$ be a TA and let $\mathcal{T} = (S, S_0, \Gamma, \delta)$ be its SG. For all $l \in L$, $u \in \mathcal{U}$ and $\tau, \tau' \in \mathbb{R}_{\ge 0}$:
\[
((l, u), \tau, (l, u+\tau)) \in \delta \land ((l, u+\tau), \tau', (l, u+\tau+\tau')) \in \delta
\]
\[
\iff ((l, u), \tau+\tau', (l, u+\tau+\tau')) \in \delta,
\]
\[
((l, u), 0, (l, u)) \in \delta.
\]
\end{property}

Within a SG, the dynamic behavior of the system is captured as a path through the state space. We define this path as an \emph{evolution}. An evolution\footnote{Throughout this paper we only consider finite length evolutions; consequently, issues related to Zeno behavior are inherently avoided.}
of a SG $\mathcal{T}=(S,S_0,\Gamma,\delta)$ is a sequence with alternating states and actions of the form
\[
\rho \;=\; s^{0}\,\gamma^{0}\,s^{1}\,\gamma^{1}\,\cdots\,\gamma^{n-1}\,s^{n} \;\in\; S\cdot(\Gamma\cdot S)^*,
\]
for some $n\in\mathbb{N}$, such that $(s^{i},\gamma^{i},s^{i+1})\in\delta$ for all $i<n$. For readability, we henceforth write evolutions in arrow form:
\[
s^0\xrightarrow{\gamma_0}s^1\xrightarrow{\gamma_1}\cdots\xrightarrow{\gamma_{n-1}}s^n.
\]
Among all possible evolutions in a system, we are specifically interested in those that represent valid system runs starting from one of the initial states. Given a SG $\mathcal{T}=(S,S_0,\Gamma,\delta)$, an evolution $\rho$ is \emph{generated} iff $s^{0}\in S_0$. We denote the set of all generated evolutions by $\mathrm{Gen}^{\mathcal{T}}\subseteq S\cdot(\Gamma\cdot S)^*$.

Finally, we formally define determinism for both the SG and the TA. Let $\mathcal{T}=(S,S_0,\Gamma,\delta)$ be a SG and $\mathcal{A}$ be a TA.
\begin{enumerate}
  \item A SG $\mathcal{T}$ is deterministic iff
  \[
  \begin{aligned}
    \forall &s\in S,\forall \gamma\in\Gamma:\ 
    \big|\{\,s'\in S \mid (s,\gamma,s')\in\delta\,\}\big|\le 1\\
    &\text{and}\quad |S_0|=1.
  \end{aligned}
  \]
  \item A TA $\mathcal{A}$ is deterministic iff its SG $\mathcal{S}(\mathcal{A})$ is a deterministic SG.
\end{enumerate}
All TA and SG considered in the remainder of this paper are deterministic. Consequently, the initial location set of any considered TA is restricted to a single initial location (typically denoted as $l_0$).

\section{The Universal Observation Model}\label{sec:observation model}
Generally, in DES, opacity requires that a system's secret behavior remains plausibly deniable to an external observer who observes system executions through a projection map (masking unobservable events)~\cite{saboori2007notions}. However, the introduction of dense time complicates this picture significantly. In real-time contexts, the precise \emph{timing} of observations serves as an additional, powerful information channel. An observer might distinguish between two sequences of identical observable events simply by analyzing the delays between them, potentially uncovering secrets that would remain hidden in an untimed abstraction~\cite{An2025}. 

First, we formalize the observer's observational capabilities. To capture the full semantics of system executions, we consider an observation model where the observer has partial visibility over locations, clock valuations, and discrete events. Specifically, each domain is partitioned into disjoint observable and unobservable subsets: $L \;=\; L_{\mathrm{obs}} \;\cup\; L_{\mathrm{uo}}, C \;=\; C_{\mathrm{obs}} \;\cup\; C_{\mathrm{uo}}, \Sigma \;=\; \Sigma_{\mathrm{obs}} \;\cup\; \Sigma_{\mathrm{uo}},$
where
\begin{enumerate}
  \item \(L_{\mathrm{obs}}\) (resp.\ \(L_{\mathrm{uo}}\)) is the set of observable (resp.\ unobservable) locations;
  \item \(C_{\mathrm{obs}}\) (resp.\ \(C_{\mathrm{uo}}\)) is the set of observable (resp.\ unobservable) clocks;
  \item \(\Sigma_{\mathrm{obs}}\) (resp.\ \(\Sigma_{\mathrm{uo}}\)) is the set of observable (resp.\ unobservable) event labels.
\end{enumerate}
Along any system evolution, we assume that time delays are fully observable. This setting reflects practical scenarios where an observer can access the wall-clock time using a local timing device. While this aligns with the premise of a globally accessible clock in standard timed-language based works (e.g., \cite{Klein2024,An2025,deng2025initial,Zhang2024, Ammar2021}), relying strictly on timed language results the \emph{suffix blindness} problem. In standard timed language-based works, any time progress occurring after the last observable event is truncated and lost in the projection. To overcome this limitation, our observation model operates over continuous state evolutions rather than event sequences, preserving these trailing delays. A formal comparison illustrating how language-based timed opacity fails under suffix blindness is discussed in Section~\ref{sec:translation lbto}.

Table \ref{tab:observer_comparison} compares our settings with other timed opacity notions. Notably, many existing works focus on event-based observations, treating locations ($L$) and clock valuations ($C$) as fully unobservable. The inclusion of $L_{\mathrm{obs}}$ and $C_{\mathrm{obs}}$ is motivated by practical scenarios where sensors provide partial information about locations and clock values. This modeling choice is conceptually aligned with the hybrid observation structures established for untimed DES, where state information is partially accessible~\cite{Ru2009,Shu2007}. By incorporating $L_{\mathrm{obs}}$ and $C_{\mathrm{obs}}$, our observation model provides a more flexible observational capability for an observer that can be reduced to these classical event-centric models under specific constraints (e.g., by setting $L_{\mathrm{obs}}=C_{\mathrm{obs}} = \varnothing$).

\begin{table}[htbp]
    \centering
    \caption{Comparison of Observer Capabilities in Timed Opacity}
    \label{tab:observer_comparison}
    \small
    \begin{tabular}{@{} l l c c c c @{}}
        \toprule
        \multirow{2}{*}{\textbf{Reference}} & \multirow{2}{*}{\textbf{Model}} & \multicolumn{4}{c}{\textbf{Observable Information}} \\
        \cmidrule(lr){3-6}
         &  & $\Sigma$ & delays & $L$ & $C$ \\
        \midrule
        
        \cite{Klein2024} & 
        discrete TA & 
        $\checkmark$ & $\checkmark$ & $\times$ & $\times$\\

        \cite{An2025} & 
        standard TA & 
        $\checkmark$ & $\checkmark$ & $\times$ & $\times$\\

        \cite{deng2025initial} & 
        integer reset TA & 
        $\checkmark$ & $\checkmark$ & $\times$ & $\times$\\

        \cite{Zhang2024} & 
        real-time labeled TA & 
        $\checkmark$ & $\checkmark$ & $\times$ & $\times$\\

        \cite{Marques2023} & 
        time interval automata & 
        $\checkmark$ & $\checkmark$ & $\times$ & $\times$\\

        \cite{Andr2023} & 
        standard TA & 
        $\times$ & $\checkmark$ & $\times$ & $\times$\\

        \cite{Ammar2021} & 
        standard TA & 
        $\checkmark$ & $\checkmark_{b}$ & $\times$ & $\times$\\

        \textbf{Ours} & 
        \textbf{standard TA} & 
        $\checkmark$ & $\checkmark$ & $\boldsymbol{\checkmark}$ & $\boldsymbol{\checkmark}$\\
        \bottomrule
    \end{tabular}
    
    \smallskip
    \footnotesize
    \raggedright
    Symbols: $\checkmark$: Observable (or Partially Observable); $\times$: Unobservable; $\checkmark_{b}$: Bounded time observation.
\end{table}

To represent information hidden from the observer, we introduce special symbols: a location symbol $l_\epsilon \notin L$, an event symbol $\sigma_\epsilon \notin \Sigma$, and a clock value $u_\epsilon \notin \mathbb{R}_{\geq 0}$. These symbols act as masks for unobservable components in the following mapping.

\begin{definition}[Observation mapping]\label{observation mapping}
Let $\Gamma = \Sigma \cup \mathbb{R}_{\geq0}$ be the set of actions, where $\Sigma=\Sigma_{\mathrm{obs}} \cup \Sigma_{\mathrm{uo}}$, $L = L_{\mathrm{obs}} \cup L_{\mathrm{uo}}$, and $C = C_{\mathrm{obs}} \cup C_{\mathrm{uo}}$ are the sets of events, locations, and clocks, respectively. 

We define the set of \emph{observable clock valuations}, denoted $\mathcal{U}_{\mathrm{obs}}$, as the set of functions mapping clocks to either a non-negative real value or the unobservable symbol: $\mathcal{U}_{\mathrm{obs}} = C \to \mathbb{R}_{\ge 0} \cup \{u_\epsilon\}.$

We define the observation mappings on locations, events, and clock valuations as $\mathcal{M}_L: L\to L_{\mathrm{obs}}\cup\{l_\epsilon\},\mathcal{M}_\Gamma: \Gamma\to \Sigma_{\mathrm{obs}}\cup\{\sigma_\epsilon\}\cup \mathbb{R}_{\geq0},$ and $\mathcal{M}_C: \mathcal U\to \mathcal{U}_{\mathrm{obs}}$ by
\[
\begin{aligned}
\mathcal{M}_L(l)=&\begin{cases} l,& l\in L_{\mathrm{obs}},\\ l_\epsilon,& l\in L_{\mathrm{uo}},\end{cases}\\
\mathcal{M}_\Gamma(\gamma)=&\begin{cases} \gamma,& \gamma\in \Sigma_{\mathrm{obs}}\cup\mathbb{R}_{\geq0},\\ \sigma_\epsilon,& \gamma\in \Sigma_{\mathrm{uo}},\end{cases}\\
\mathcal{M}_C(u)(c)= 
&\begin{cases} 
u(c), & c \in C_{\mathrm{obs}}, \\
u_\epsilon, & c \in C_{\mathrm{uo}}.\qquad\qquad\qquad\qquad\quad\quad\qedhere
\end{cases}
\end{aligned}
\]
\end{definition}
To evaluate the time elapse operation on observable clock valuations, we extend the addition operation to the domain $\mathcal{U}_{\mathrm{obs}}$. For any $u \in \mathcal{U}_{\mathrm{obs}}$ and $\tau \in \mathbb{R}_{\geq 0}$, the valuation $u + \tau$ is defined as:\[(u + \tau)(c) = 
\begin{cases} 
    u(c) + \tau, & \text{if } u(c) \in \mathbb{R}_{\geq 0}, \\
    u_\epsilon, & \text{if } u(c) = u_\epsilon.
\end{cases}\]

Building upon Definition~\ref{observation mapping}, we can naturally extend the observation mapping to evolutions. Let $\mathcal{A}=(L,l_0,C,\Sigma,E,I)$ be a TA and let 
$\mathcal{T}=\mathcal{S}(\mathcal{A})=(S,S_0,\Gamma,\delta)$ be its SG (with $S= L\times\mathcal U$ and $\Gamma=\Sigma\cup\mathbb{R}_{\ge0}$).
Let $\mathcal{M}_\rho:\ S\cdot(\Gamma\cdot S)^*\ \to\ S^{o}\cdot(\Gamma^{o}\cdot S^{o})^*$ denote the observation mapping on evolutions,
where
\[
S^{o} := (L_{\mathrm{obs}}\cup\{l_\epsilon\})\times \mathcal{U}_{\mathrm{obs}},
\qquad
\Gamma^{o} := \Sigma_{\mathrm{obs}}\cup\{\sigma_\epsilon\}\cup\mathbb{R}_{\ge0}.
\]
We first define the observation of a single state $s=(l,u)\in S$ as the component-wise projection $\mathcal{M}_{\mathfrak{S}}(s) = (\mathcal{M}_L(l), \mathcal{M}_C(u))$. The mapping $\mathcal{M}_\rho$ is then defined inductively on the structure of evolutions:
\begin{itemize}
    \item Base case: For an evolution of length\footnote{We measure the length of an evolution by the number of actions. For each $s\in S$, the sequence $(s)\in S\cdot(\Gamma\cdot S)^*$ has length $0$. We treat $S$ and $S\cdot(\Gamma\cdot S)^*$ as disjoint; in particular, $(s)\neq s$. For clarity, evolutions of positive length continue to be written \emph{without} parentheses, e.g., $s_0\cdot\gamma_0\cdot s_1\cdot\cdots$, and we reserve parentheses \emph{only} for the zero-length case. 
} 0 consisting of a single state $s \in S$, $\mathcal{M}_\rho((s)) = \mathcal{M}_{\mathfrak{S}}(s).$ 
    \item Inductive step: For an evolution $\rho = \rho' \cdot \gamma \cdot s$, where $\rho'$ is another evolution, $\gamma \in \Gamma$ is an action, and $s \in S$ is the resulting state:
    \[
    \mathcal{M}_\rho(\rho' \cdot \gamma \cdot s) = \mathcal{M}_\rho(\rho') \cdot \mathcal{M}_\Gamma(\gamma) \cdot \mathcal{M}_{\mathfrak{S}}(s).
    \]
\end{itemize}


To lighten notation, in what follows we write $\mathcal{M}$ in place of $\mathcal{M}_\rho$ whenever the argument is an evolution. 
The component-wise mappings keep their subscripts $\mathcal{M}_L$, $\mathcal{M}_\Gamma$, and $\mathcal{M}_C$.

However, the observation mapping $\mathcal{M}$ alone is insufficient to fully capture the observer's perspective. In a dense-time setting, the observer cannot distinguish between a continuous delay and a fragmented delay interrupted by unobservable events. Therefore, we define an equivalence relation that normalizes these variations, treating semantically distinct but observationally identical behaviors as equivalent.

\begin{definition}[Observational equivalence]\label{def:obs-equivalence}
Let $\mathcal{A}=(L,l_0,C,\Sigma,E,I)$ be a TA and $\mathcal T=(S,s_0,\Gamma,\delta)$ be its SG. We define two elementary equivalence relations on evolutions based on left contexts $\alpha\in (S\cdot\Gamma)^*\ \cup\ (S^{o}\cdot\Gamma^{o})^*$ and right contexts $\beta\in (\Gamma\cdot S)^*\ \cup\ (\Gamma^{o}\cdot S^{o})^*$:

\begin{enumerate}
    \item Silent-step equivalence ($\equiv_\epsilon$):
    Defined as the smallest equivalence relation satisfying the following condition which identifies evolutions where $\sigma_\epsilon$-transitions do not result in a change of the original state or observed state. For any state $s \in S \cup S^o$:
    \[\alpha\, s \xrightarrow{\sigma_\epsilon} s\, \beta \;\;\equiv_\epsilon\;\;\alpha\,s\,\beta.\]
    \item Time additivity, continuity, and reflexivity equivalence ($\equiv_\tau$):
    Defined as the smallest equivalence relation satisfying the following conditions which identify evolutions that differ by the fragmentation of delays or zero delay. For any state $(l,u)\in S \cup S^o$ and delays $\tau, \tau' \in \mathbb{R}_{\geq 0}$:
    \[
    \begin{aligned}
    \alpha\,(l,u)\xrightarrow{\tau}(l,u{+}\tau)\xrightarrow{\tau'}(l,u{+}\tau{+}\tau')\,&\beta
    \;\;\equiv_\tau\;\;\\
    \alpha\,(l,u)\xrightarrow{\tau{+}\tau'}(l,u{+}\tau{+}\tau')\,&\beta,\\
    \alpha\, s \xrightarrow{0} s\, \beta \;\;\equiv_\tau\;\; \alpha\,s\,\beta.    
    \end{aligned}
    \]
\end{enumerate}

We say two evolutions $\rho$ and $\rho'$ are \emph{observably equivalent}, denoted by $\rho \equiv_{\mathrm{obs}} \rho'$, if they belong to the reflexive and transitive closure of the union of these two relations, i.e., $(\equiv_\epsilon \cup \equiv_\tau)^*$.

The equivalence class of an evolution $\rho$ under this relation is denoted by $\llbracket \rho \rrbracket_{\equiv_{\mathrm{obs}}}$.
\end{definition}
\section{Notion of EBTO}\label{sec:notion of EBTO}
With the observation model established, we can now define what it means for a system to be opaque. We introduce \emph{EBTO}, which protects the secrecy of entire execution histories. Throughout this paper, the observation mapping $\mathcal{M}$ is induced by the specific configuration of observable sets $(L_{\mathrm{obs}}, C_{\mathrm{obs}}, \Sigma_{\mathrm{obs}})$. To focus on the core mathematical properties of opacity, we omit explicit mention of these sets in subsequent definitions and theorems unless a change in the observation model (e.g., setting $L_{\mathrm{obs}} = \varnothing$) is specifically discussed.

\begin{definition}[EBTO]\label{def:ebto}
Let $\mathcal{T}=(S,S_0,\Gamma,\delta)$ be a SG. A subset $\mathrm{Gen}^{\mathcal{T}}_s\subseteq \mathrm{Gen}^{\mathcal{T}}$ is a \emph{secret set} if it is closed under time additivity, continuity, and reflexivity equivalence:
\[
\forall \rho\in \mathrm{Gen}^{\mathcal{T}}_s,\quad \llbracket \rho \rrbracket_{\equiv_{\tau}} \subseteq \mathrm{Gen}^{\mathcal{T}}_s.
\]
The system $\mathcal{T}$ is said to be \emph{evolution-based timed opaque} w.r.t. the secret set $\mathrm{Gen}^{\mathcal{T}}_s$ and the observation mapping $\mathcal{M}$ iff $\forall \rho\in \mathrm{Gen}^{\mathcal{T}}_s,\; \exists \rho'\in \mathrm{Gen}^{\mathcal{T}}\setminus \mathrm{Gen}^{\mathcal{T}}_s
\text{ such that } \llbracket\mathcal{M}_\rho(\rho) \rrbracket_{\equiv_{\mathrm{obs}}}=\llbracket\mathcal{M}_\rho(\rho') \rrbracket_{\equiv_{\mathrm{obs}}}.$

A TA $\mathcal{A}$ is EBTO iff its SG is EBTO.
\end{definition}

The closure of $\mathrm{Gen}_s^{\mathcal T}$ under $\equiv_\tau$ is strictly motivated by the physical reality of observation, rather than mathematical convenience. In a physical system, an observer measuring time cannot distinguish between a single, continuous delay of $\tau + \tau'$ and the exact same delay conceptually fragmented into $\tau$ followed by $\tau'$. Similarly, a zero-delay transition ($\xrightarrow{0}$) takes no actual time and produces no physical progression. Consequently, a secret execution could be ``hidden" behind a non-secret execution that differs only on paper, e.g., by arbitrarily splitting a delay or inserting a physically meaningless zero-delay step. This would render the opacity condition vacuously true, allowing a system to be deemed secure by exploiting a mathematical loophole rather than reflecting true physical opacity.

To illustrate why this closure constraint is necessary, let us assume that the secret set is \emph{not} required to satisfy this closure property. Under this assumption, consider a finite SG as shown in Fig.~\ref{nomalizition_closed}, illustrating states reachable from $s_0$ via time-delay transitions. Suppose we define the single-step evolution $s_0 \xrightarrow{t_1+t_2} s_2$ as secret. At the same time, because we are ignoring the closure constraint, we are allowed to classify the fragmented step-by-step evolution $s_0 \xrightarrow{t_1} s_1 \xrightarrow{t_2} s_2$, as well as the evolution with a zero-delay $s_0 \xrightarrow{0} s_0 \xrightarrow{t_1+t_2} s_2$, as non-secret.
\begin{figure}[t]
  \centering
  \scalebox{0.7}{
  \begin{tikzpicture}[->, >=stealth, shorten >=1pt, node distance=2cm, on grid, auto]
    \node[state, initial, initial text=] (s0) {$s_0$};
    \node[state, right of=s0] (s0d) {$s_1$};
    \node[state, right of=s0d] (s0dd) {$s_2$};

    \path (s0) edge node[below] {$t_1$} (s0d);
    \path (s0d) edge node[below] {$t_2$} (s0dd);

    \path (s0) edge[bend left=40] node[above] {$t_1+t_2$} (s0dd);
    
    \path (s0) edge[loop below] node[below] {$0$} (s0);
    \path (s0d) edge[loop below] node[below] {$0$} (s0d);
    \path (s0dd) edge[loop below] node[below] {$0$} (s0dd);

  \end{tikzpicture}}
  \caption{Semantically indistinguishable time-delay evolutions}
  \label{nomalizition_closed}
\end{figure}
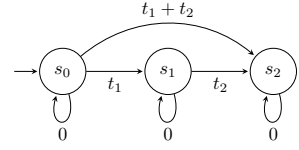
In this case, although such a setting is valid mathematically, it provides limited analytical value: regardless of how we choose the observable location set $L_{\mathrm{obs}}$, the observable action set $\Sigma_{\mathrm{obs}}$, or the observable clock set $C_{\mathrm{obs}}$, these three evolutions remain completely indistinguishable under the standard semantics of time delay.

\section{Translation Between Opacity Notions}\label{sec:translation}
In this section, we establish a hierarchy of timed opacity notions based on their expressiveness. We demonstrate that the proposed framework serves as a strictly more general semantic baseline by establishing constructive reductions from existing notions, such as LBTO~\cite{An2025} and ETO~\cite{Andr2023}, to EBTO, and we formally prove that the converse does not hold in general.
\subsection{Translation Between EBTO and LBTO}\label{sec:translation lbto}
In LBTO~\cite{An2025}, system behaviors are characterized by \emph{timed words}. To formalize the observer's observational capabilities in this setting, we recall the standard definition of a timed word and its projection. 

Let $\Sigma$ be a finite set of events. A \emph{timed word}~\cite{alur1994theory} over $\Sigma$ is a finite sequence of pairs
\[
\omega = (\sigma_1, t_1)(\sigma_2, t_2)\cdots(\sigma_n, t_n) \in (\Sigma \times \mathbb{R}_{\geq 0})^n, n\in\mathbb{N},
\]
where each $\sigma_i \in \Sigma$ is an event, and $t_i \in \mathbb{R}_{\geq 0}$ is a global timestamp such that $0 \leq t_1 \leq t_2 \leq \cdots \leq t_n$. A \emph{timed language} over $\Sigma$ is a set of timed words over $\Sigma$, i.e., $\mathcal{L} \subseteq (\Sigma \times \mathbb{R}_{\geq 0})^*$.

Given a subset $ \Sigma_{\mathrm{obs}} \subseteq \Sigma $, a \textit{projection} $ P_{\Sigma_{\mathrm{obs}}} $ on timed words w.r.t. $ \Sigma_{\mathrm{obs}} $ is a function
\[
P_{\Sigma_{\mathrm{obs}}}: (\Sigma \times \mathbb{R}_{\geq 0})^* \to (\Sigma_{\mathrm{obs}} \times \mathbb{R}_{\geq 0})^*
\]
such that
\[
\begin{aligned}
P_{\Sigma_{\mathrm{obs}}}(\epsilon) &= \epsilon\\
P_{\Sigma_{\mathrm{obs}}}(\omega \cdot (\sigma,t))& =
\begin{cases}
    P_{\Sigma_{\mathrm{obs}}}(\omega) \cdot (\sigma,t) & \text{if } \sigma \in \Sigma_{\mathrm{obs}} \\
    P_{\Sigma_{\mathrm{obs}}}(\omega) & \text{otherwise.}
\end{cases}
\end{aligned}
\]
Note that this definition structurally leads to the \emph{suffix blindness} mentioned earlier, as any time progress occurring after the last observable event is discarded in the projected timed word.

Next, we establish a formal reduction from LBTO to EBTO. The key difference between them is that LBTO lacks any state information: EBTO evaluates semantic evolutions containing locations and clocks, whereas LBTO relies solely on timed words. To bridge this structural gap, we define a mapping $\mathcal{Y}$ that projects detailed system executions onto observable timed sequences by accumulating continuous delays into global timestamps.
\begin{definition}[Mapping evolutions to timed words]\label{def:Y-mapping}
Let $\mathcal{T} = (S, s_0, \Gamma, \delta)$ be a SG with actions $\Gamma = \Sigma \cup \mathbb{R}_{\geq 0}$.
Consider a generated evolution $\rho \;=\; s_0 \xrightarrow{\gamma_0} s_1 \xrightarrow{\gamma_1} \cdots \xrightarrow{\gamma_{n-1}} s_n.$

We define the mapping $\mathcal{Y}$ that converts $\rho$ into a timed word by projecting the sequence onto  $(\Sigma \times \mathbb{R}_{\geq 0})^*$. Let $0 \le k_1 < k_2 < \dots < k_m < n$ be the ordered indices of the transitions that correspond to discrete events, i.e., $\gamma_{k_j} \in \Sigma$.
The \emph{timed word} of $\rho$ is defined as:
\[
\mathcal{Y}(\rho) \;=\; (\gamma_{k_1}, t_{k_1})\,(\gamma_{k_2}, t_{k_2})\,\cdots\,(\gamma_{k_m}, t_{k_m}),
\]
where each global timestamp $t_{k_j}$ is the sum of all preceding time delays:
\[
t_{k_j} \;=\; \sum_{h = 0}^{k_j - 1} \mathrm{dur}(\gamma_h),
\quad \text{with} \quad
\mathrm{dur}(\gamma) = 
\begin{cases} 
\gamma, & \text{if } \gamma \in \mathbb{R}_{\geq 0}, \\ 
0, & \text{if } \gamma \in \Sigma. 
\end{cases}
\]

The inverse mapping $\mathcal{Y}^{-1}$ assigns to a timed word $\omega$ the set of all evolutions generating it:
\[
\mathcal{Y}^{-1}(\omega) \;=\; \bigl\{\, \rho \in \mathrm{Gen}^{\mathcal{T}} \mid \mathcal{Y}(\rho) = \omega \,\bigr\}. \qedhere
\]
\end{definition}

We naturally extend these mappings to sets. For a set of evolutions $\mathrm{Evo} \subseteq \mathrm{Gen}^{\mathcal{T}}$, its generated language is $\mathcal{Y}(\mathrm{Evo}) = \{\mathcal{Y}(\rho) \mid \rho \in \mathrm{Evo}\}$. Conversely, for a given timed language $\mathcal{L}$, its preimage is the set of all evolutions that generate a word in $\mathcal{L}$:
\[
\mathcal{Y}^{-1}(\mathcal{L}) \;=\; \bigcup_{\omega \in \mathcal{L}} \mathcal{Y}^{-1}(\omega) \;=\; \{\, \rho \in \mathrm{Gen}^{\mathcal{T}} \mid \mathcal{Y}(\rho) \in \mathcal{L} \,\}.
\]
For instance, let $\rho$ be a valid evolution given by:
\[
\rho \;=\; s_0 \xrightarrow{1.5} s_1 \xrightarrow{a} s_2 \xrightarrow{0.7} s_3 \xrightarrow{0.3} s_4 \xrightarrow{b} s_5 \xrightarrow{0.2} s_6,
\]
the resulting timed word is: $\mathcal{Y}(\rho) \;=\; (a, 1.5)\,(b, 2.5)$.

This example illustrates how the mapping abstracts away intermediate states ($s_1, s_3, s_4$) and merges multiple delay steps into a single global timestamp. Notably, much like the projection $P_{\Sigma_{\mathrm{obs}}}$ discussed earlier, this mapping discards any time delay occurring after the final discrete event, e.g., the $0.2$ delay leading to $s_6$.

In the context of LBTO, the focus lies on the generated language of the system. To this end, we define the \emph{generated language} of a TA $\mathcal{A}$, denoted by $\mathcal{L}_{\mathrm{gen}}(\mathcal{A})$, by applying the mapping $\mathcal{Y}$ to all generated evolutions $\mathcal{L}_{\mathrm{gen}}(\mathcal{A}) \;=\; \mathcal{Y}(\mathrm{Gen}^{\mathcal{T}}).$ Then, we recall the definition of LBTO.

\begin{definition}[Language-based timed opacity, LBTO~\cite{An2025}]\label{def:lbto}
A TA $\mathcal{A}$ is language-based timed opaque w.r.t $\Sigma_{\mathrm{obs}}$ and the secret language $\mathcal{L}_s$ iff $\forall \omega \in \mathcal{L}_{\mathrm{gen}}(\mathcal{A}) \cap \mathcal{L}_s, \exists \omega' \in \mathcal{L}_{\mathrm{gen}}(\mathcal{A}) \setminus \mathcal{L}_s \text{ s.t. } P_{\Sigma_{\mathrm{obs}}}(\omega) = P_{\Sigma_{\mathrm{obs}}}(\omega')$.
\end{definition}

\begin{definition}[Conversion function, LBTO to EBTO]\label{CLE}
    Given a TA $\mathcal{A}=(L,l_0,C,\Sigma,E,I)$ and a secret language $\mathcal{L}_s$, the conversion function is $\mathcal{C}_{\mathrm{LE}}(\mathcal{A},\mathcal{L}_s)=(\mathcal{A}, \mathrm{Gen}^{\mathcal{T}}_s)$, where the set of secret evolutions is the exact inverse image of the secret language:
    \[
    \mathrm{Gen}_s^{\mathcal{T}} \;=\; \mathcal{Y}^{-1}(\mathcal{L}_s). \qedhere
    \]
\end{definition}

Furthermore, a valid reduction requires formal consistency between the observation models of both domains. Since LBTO restricts the observer to discrete events and timestamps, the EBTO observer must be restricted from accessing location and clock data (e.g., $L_{\mathrm{obs}} = \varnothing$ and $C_{\mathrm{obs}} = \varnothing$). Under this specific restriction, observational equivalence in the SG directly implies projection equality in the language domain. That is, for any $\rho, \rho' \in \mathrm{Gen}^{\mathcal{T}}$:
\[
\begin{aligned}
   & \llbracket\mathcal{M}(\rho)\rrbracket_{\equiv_{\mathrm{obs}}} = \llbracket\mathcal{M}(\rho')\rrbracket_{\equiv_{\mathrm{obs}}} \\&\implies P_{\Sigma_{\mathrm{obs}}}(\mathcal{Y}(\rho)) = P_{\Sigma_{\mathrm{obs}}}(\mathcal{Y}(\rho')).
\end{aligned}
\](The formal description and proof are provided in \hypertarget{main:lem:obs_equiv}{\hyperref[app:sec:lem:obs_equiv]{Appendix~\ref*{app:sec:lem:obs_equiv}}}.)

These structural and observational equivalences directly yield the main reduction theorem.

\begin{theorem}[Sufficiency of EBTO for LBTO]\label{thm:lbto_to_ebto}
Let $\mathcal{A}$ be a TA with SG $\mathcal{T}=\mathcal{S}(\mathcal{A})$, let $\Sigma_{\mathrm{obs}} \subseteq \Sigma$ be a given set of observable events where locations and clocks are unobservable (i.e., $L_{\mathrm{obs}} = \varnothing$ and $C_{\mathrm{obs}} = \varnothing$), and let $\mathcal{L}_s$ be a secret language. Let $(\mathcal{A}, \mathrm{Gen}_s^{\mathcal{T}}) = \mathcal{C}_{\mathrm{LE}}(\mathcal{A}, \mathcal{L}_s)$.

The TA $\mathcal{A}$ is LBTO w.r.t. $\Sigma_{\mathrm{obs}}$ and $\mathcal{L}_s$ if it is EBTO w.r.t. the $\mathcal{M}$ and $\mathrm{Gen}_s^{\mathcal{T}}$.
\end{theorem}
\noindent\textit{Proof.} The formal proof is provided in \hypertarget{main:thm:lbto_to_ebto}{\hyperref[app:thm:lbto_to_ebto]{Appendix~\ref*{app:sec:thm:lbto_to_ebto}}}.

Theorem~\ref{thm:lbto_to_ebto} establishes a strict one-way implication. The reverse implication does not hold due to a fundamental structural limitation in language-based formulations, which we term \emph{suffix blindness}. Specifically, the LBTO observation mapping truncates the continuous time progress at the occurrence of the last observable event, discarding any trailing time delays or subsequent unobservable state evaluations. A TA might be deemed secure under LBTO simply because this language projection masks terminal timing discrepancies. In contrast, the observation equivalence $\equiv_{\mathrm{obs}}$ in EBTO strictly preserves the continuous progress of time, identifying evolutions that diverge only in their trailing delays.

Consider the TA shown in Fig.~\ref{fig:suffix blindness example}, where event $a$ is observable and event $b$ is unobservable. The TA generates two timed words: $\omega_1 = (a, 1)$ and $\omega_2 = (a, 1)(b, 100)$. Let the secret language be $\mathcal{L}_s = \{\omega_1\}$ and non-secret language be $\{\omega_2\}$. Under the LBTO projection, the unobservable event $b$ and its timestamp are discarded. Thus, $P_{\Sigma_{\mathrm{obs}}}(\omega_1) = P_{\Sigma_{\mathrm{obs}}}(\omega_2) = (a, 1)$, the TA is LBTO. 

However, under EBTO, the set of secret evolutions is defined as $\mathrm{Gen}_s^{\mathcal{T}} = \mathcal{Y}^{-1}(\{\omega_1\})$, which contains all evolutions with a total time delay $t \in [1, 100)$. The set of non-secret evolutions $\mathcal{Y}^{-1}(\{\omega_2\})$ contains evolutions with $t \ge 100$. Since no secret evolution is observationally equivalent to any non-secret evolution, the observer can identify the secret. Therefore, the TA is not EBTO. This example shows that LBTO fails to detect the secret because its projection discards trailing delays. In contrast, EBTO preserves these delays and identifies the system as not opaque.

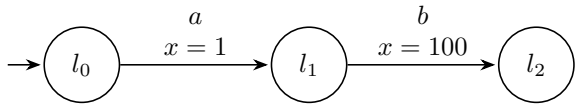
\begin{figure}[b]
    \centering
    \begin{tikzpicture}[
        ->,                      
        >=Stealth,               
        shorten >=1pt,           
        auto,                    
        node distance=2cm,       
        semithick,               
        state/.style={circle, draw, minimum size=1cm} 
    ]
        \node[state, initial, initial text={}] (l0) {$l_0$};
        \node[state] (l1) [right=of l0] {$l_1$};
        \node[state] (l2) [right=of l1] {$l_2$};

        \path (l0) edge node[align=center] {$a$\\ $x=1$} (l1)
              (l1) edge node[align=center] {$b$\\ $x=100$} (l2);

    \end{tikzpicture}
    \caption{A TA example illustrating suffix blindness.}
    \label{fig:suffix blindness example}
\end{figure}

To formally demonstrate this asymmetry, the following proposition proves that EBTO is strictly more expressive, as no LBTO setting can capture secrets embedded exclusively within terminal delays.

\begin{proposition}[Strict expressiveness of EBTO over LBTO]\label{prop:ebto_to_lbto}
There exists a TA $\mathcal{A}$ and a set of secret evolutions $\mathrm{Gen}_s^{\mathcal{T}}$ such that no LBTO secret language $\mathcal{L}_s$ over any observable event set can exactly capture the same secret behaviors. Specifically, $\mathrm{Gen}_s^{\mathcal{T}} \neq \mathcal{Y}^{-1}(\mathcal{L}_s)$ for all possible $\mathcal{L}_s$.
\end{proposition}
\noindent\textit{Proof.} The formal proof is provided in \hypertarget{main:prop:ebto_to_lbto}{\hyperref[app:prop:ebto_to_lbto]{Appendix~\ref*{app:sec:prop:ebto_to_lbto}}}.

\subsection{Translation Between EBTO and ETO}
Unlike LBTO, which relies on event sequences, ETO~\cite{Andr2023} is formulated over \emph{runs} where consecutive time delays are merged (by time additivity) and evaluated based on their accumulated duration.

\begin{definition}[ETO~\cite{Andr2023}]\label{def:ETO_original}
Let $\mathcal{A}=(L, l_0, C, \Sigma, E, I)$ be a TA. Let $l_{\mathrm{priv}} \in L$ be a private location and $l_f \in L$ be a final location.
A \emph{run} $\mathsf {Run}$ of $\mathcal{A}$ is an alternating sequence of concrete states and pairs of delays and events starting from the initial state $(l_0, \mathbf{0})$:
\[
\begin{aligned}
    \mathsf {Run} = (l_0, \mathbf{0}), &(d_0, e_0), (l_1, u_1), \cdots,(d_{m-1}, e_{m-1}), (l_m, u_m),\\&\cdots, (d_{n-1}, e_{n-1}), (l_n, u_n),\quad m,n\in\mathbb{N},
\end{aligned}
\]
where $d_i \in \mathbb{R}_{\geq 0}$ represents a delay, and $e_i \in \Sigma$ represents a discrete event.

The \emph{duration} of a $\mathsf {Run}$ is defined as the sum of delays:
\[
\mathsf {Duration}(\mathsf {Run}) = \sum_{0 \leq i \leq n-1} d_i.
\]

We define two sets of runs based on the exact visitation of specific locations, $l_{\mathrm{priv}}$ and $l_f$:

\begin{enumerate}
    \item Private runs $\mathsf{Visit^{\mathrm{priv}}}(\mathcal{A})$: The set of runs such that
    \[
    l_m = l_{\mathrm{priv}} \land l_n = l_f \land (\forall 0 \leq i \leq n-1, l_i \neq l_f).
    \]
    This formally captures runs that visit $l_{\mathrm{priv}}$ on the way to the first occurrence of $l_f$. We denote by $\mathsf{DVisit}^{\mathrm{priv}}(\mathcal{A})$ the set of all durations of these runs.

    \item Public runs $\overline{\mathsf{{Visit}}}^{\mathrm{priv}}(\mathcal{A})$: The set of runs such that 
    \[\begin{aligned}
    l_n = l_f \land (\forall 0 \leq i \leq n-1, l_i \neq l_f) \land \\(\forall 0 \leq j \leq n, l_j \neq l_{\mathrm{priv}}).
    \end{aligned}
    \]
    This captures runs that reach $l_f$ without ever visiting $l_{\mathrm{priv}}$. We denote by $\mathsf{D\overline{Visit}^{\mathrm{priv}}}(\mathcal{A})$ the set of all durations of these runs.
\end{enumerate}

The system $\mathcal{A}$ is said to be \emph{execution-time opaque}\footnote{We focus on the notion of \emph{weak ET-opacity} as defined in \cite{Andr2023}.} (specifically, Weak ETO) if the set of private durations is included in the set of public durations:
\[
\mathsf{DVisit}^{\mathrm{priv}}(\mathcal{A}) \subseteq \mathsf{D\overline{Visit}^{\mathrm{priv}}}(\mathcal{A}).\qedhere
\]
\end{definition}

To analyze ETO within our proposed framework, we define a normalization mapping. 

\begin{definition}[Normalization mapping]\label{def:normalization_mapping}
We define the normalization mapping $\Psi: \mathrm{Gen}^{\mathcal{T}} \to Runs(\mathcal{A})$. 
For any evolution $\rho \in \mathrm{Gen}^{\mathcal{T}}$, let $0 \le k_1 < k_2 < \dots < k_m < n$ be the indices of all discrete events in $\rho$ (i.e., $\gamma_{k_j} \in \Sigma$).
The normalized run is the sequence of $m$ steps:
\[\begin{aligned}
\Psi(\rho) \;=\; &s_0,\, (d_1, \gamma_{k_1}),\, s_{k_1+1},\, (d_2, \gamma_{k_2}),\, s_{k_2+1},\\&\, \dots,\, (d_m, \gamma_{k_m}),\, s_{k_m+1},
\end{aligned}
\]
where for each $j \in \{1, \dots, m\}$, the delay $d_j$ is the sum of all delays between the previous discrete event $\gamma_{k_{j-1}}$ (or the start of the evolution if $j=1$) and the current event $\gamma_{k_j}$.

The reversed function is defined as:
\[
\Psi^{-1}(\mathsf {Run} ) \;=\; \{\, \rho \in \mathrm{Gen}^{\mathcal{T}} \mid \Psi(\rho) = \mathsf {Run}  \,\}. \qedhere
\]
\end{definition}
Note that if $\rho$ ends with a time delay, i.e., there are time delays after the last discrete event $\gamma_{k_m}$, this trailing duration is  not included in $\Psi(\rho)$.

Consider a SG $\mathcal{T}$ and a generated evolution $\rho$ involving states $s_0, \dots, s_6$:
\[
\rho \;=\; s_0 \xrightarrow{0.5} s_1 \xrightarrow{1.5} s_2 \xrightarrow{a} s_3 \xrightarrow{3.0} s_4 \xrightarrow{b} s_5 \xrightarrow{4.5} s_6,
\]
the resulting canonical run is:
\[
\Psi(\rho) \;=\; s_0, (2.0, a), s_3, (3.0, b), s_5.
\]
This run has a total duration of $5.0$. Note that the semantic duration of $\rho$ was $9.5$. This discrepancy reveals a structural divergence: while a general evolution in a SG can terminate at any arbitrary time point, an ETO run only accounts for the time accumulated up to the final discrete event. In ETO, the duration of a run accounts for the time accumulated up to the first arrival at the final location $l_f$. To bridge this with our framework, we define the set of evolutions with this termination condition. Let $\mathrm{Gen}_{l_f}^{\mathcal{T}} \subseteq \mathrm{Gen}^{\mathcal{T}}$ denote the set of all generated evolutions that terminate exactly upon entering $l_f$ for the first time. Because location changes in TA are strictly driven by discrete events, every evolution in $\mathrm{Gen}_{l_f}^{\mathcal{T}}$ naturally ends with a discrete transition. Consequently, there are no unrecorded trailing delays, and the total duration in the SG corresponds exactly to the run duration in the ETO framework: $\forall \rho \in \mathrm{Gen}_{l_f}^{\mathcal{T}}, \mathsf{Duration}(\rho) = \mathsf{Duration}(\Psi(\rho))$.

\begin{definition}[Conversion function, ETO to EBTO]\label{def:C_ExE}
    Given a TA $\mathcal{A}$ with a $l_{\mathrm{priv}}$ and $l_f$, we define the conversion function $\mathcal{C}_{\mathrm{ExE}}(\mathcal{A},l_{\mathrm{priv}},l_f) = (\mathcal{A}, \mathrm{Gen}_s^{\mathcal{T}})$. The set of secret evolutions is defined as the subset of evolutions ending at $l_f$ that map to a private run:
    \[
    \mathrm{Gen}_s^{\mathcal{T}} = \{ \rho \in \mathrm{Gen}_{l_f}^{\mathcal{T}} \mid \Psi(\rho) \in \mathsf{Visit^{\mathrm{priv}}}(\mathcal{A})\}.\qedhere
    \]
\end{definition}

A valid reduction requires consistency between the information available to the observer in both domains. Setting $L_{\mathrm{obs}} = \{l_f\}$, $\Sigma_{\mathrm{obs}} = \varnothing$, and $C_{\mathrm{obs}} = \varnothing$ ensures the EBTO observer is blind to all intermediate states and discrete events, matching the duration-centric focus of ETO. Under this strict observation model, the observer only perceives the continuous flow of time culminating in the arrival at the final location $l_f$. Consequently, for any two generated evolutions $\rho, \rho' \in \mathrm{Gen}_{l_f}^{\mathcal{T}}$, their observational equivalence holds if and only if their total durations are equal:
\[\begin{aligned}
\llbracket \mathcal{M}(\rho) \rrbracket_{\equiv_{\mathrm{obs}}} = \llbracket \mathcal{M}(\rho') \rrbracket_{\equiv_{\mathrm{obs}}} \iff \\\mathsf{Duration}(\rho) = \mathsf{Duration}(\rho').
\end{aligned}
\]\hypertarget{main:lem:duration_strict}{(The formal description and proof are provided in \hyperref[lem:duration_strict]{Appendix~\ref*{app:sec:lem:duration_strict}}).}

\begin{theorem}\label{thm:eto_to_ebto}
Let $\mathcal{A}$ be a TA with private location $l_{\mathrm{priv}}$ and final location $l_f$. Let $L_{\mathrm{obs}}=\{l_f\}$, $\Sigma_{\mathrm{obs}} = \varnothing$, and $C_{\mathrm{obs}} = \varnothing$. Let $(\mathcal{A}, \mathrm{Gen}_s^{\mathcal{T}}) = \mathcal{C}_{\mathrm{ExE}}(\mathcal{A},l_{\mathrm{priv}},l_f)$ be the converted system. 

The TA $\mathcal{A}$ is ETO w.r.t. $l_{\mathrm{priv}},l_f$ iff it is EBTO w.r.t. the observation mapping $\mathcal{M}$ and secret evolutions $\mathrm{Gen}_s^{\mathcal{T}}$.
\end{theorem}

\noindent\textit{Proof.} The formal proof is provided in \hypertarget{main:thm:eto_to_ebto}{\hyperref[thm:eto_to_ebto]{Appendix~\ref*{app:sec:thm:eto_to_ebto}}}.

Theorem~\ref*{thm:eto_to_ebto} demonstrates that ETO is structurally equivalent to EBTO under a minimalist observation model where all discrete components, locations, clocks, and events, are unobservable (i.e., $\Sigma_{\mathrm{obs}} = L_{\mathrm{obs}} = C_{\mathrm{obs}} = \varnothing$).

\begin{proposition}[Strict expressiveness of EBTO over ETO]\label{prop:ebto_to_eto}
There exists an EBTO instance $(\mathcal{A}, \mathrm{Gen}_s^{\mathcal{T}})$ such that for any TA $\mathcal{A}'$ and $l_{\mathrm{priv}}$ and $l_f$, the ETO property of $(\mathcal{A}', l_{\mathrm{priv}}$ and $l_f)$ is not equivalent to the EBTO property of $(\mathcal{A}, \mathrm{Gen}_s^{\mathcal{T}})$.
\end{proposition}

\noindent\textit{Proof.} The formal proof is provided in \hypertarget{main:prop:ebto_to_eto}{\hyperref[prop:ebto_to_eto]{Appendix~\ref*{app:sec:prop:ebto_to_eto}}}.

The fundamental barrier to this reduction lies in the information loss inherent in the ETO measurement mechanism. ETO characterizes system behavior through a single scalar value, the accumulated duration of a completed run. In contrast, EBTO preserves the full temporal trajectory of the evolution.

In this section, we established a formal hierarchy of timed opacity classes. We demonstrated that EBTO serves as a unified semantic baseline capable of capturing both language-based observations (LBTO) and duration-based observations (ETO). Crucially, we provided constructive transformations showing that, under strictly constrained observation models, both event-centric and duration-centric formulations are strict instantiations of the proposed evolution-based framework. This systematic unification, illustrated in Fig.~\ref{fig:relations}, confirms that diverse security requirements can be rigorously analyzed within a single cohesive methodology.

\begin{figure}[b]
\centering
\scalebox{0.8}{
\begin{tikzpicture}[
    >=stealth,
    node distance=2.2cm and 2.8cm, 
    tac_box/.style={
        rectangle,
        draw=black,
        thick,
        minimum width=2.4cm, 
        minimum height=1.0cm, 
        align=center,
        font=\small\bfseries, 
        fill=white
    },
    proposed_box/.style={
        tac_box,
        fill=black!8 
    },
    existing_box/.style={
        tac_box
    },
    arrow_line/.style={->, thick, draw=black!80, shorten >=2pt, shorten <=2pt}, 
    arrow_dashed/.style={->, dashed, thick, draw=black!80, shorten >=2pt, shorten <=2pt},
    label_text/.style={font=\scriptsize, text=black!90, inner sep=2pt} 
]

    
    \node[proposed_box] (EBTO) {EBTO};
    
    \node[existing_box, below=1.5cm of EBTO, xshift=-2.2cm] (LBTO) {LBTO};
    \node[existing_box, below=1.5cm of EBTO, xshift=2.2cm] (ETO) {ETO};

    
    \draw[arrow_line] ([xshift=-0.8cm+5pt]EBTO.south) -- 
        node[right, label_text, pos=0.7] {Sufficiency} ([xshift=5pt]LBTO.north);

    \draw[arrow_line]   ([xshift=5pt]ETO.north) -- 
        node[right, label_text, pos=0.3] {Reduction} ([xshift=0.8cm+5pt]EBTO.south);
    \draw[arrow_dashed] ([xshift=0.8cm-5pt]EBTO.south) -- 
        node[left, label_text, pos=0.7] {Restriction} ([xshift=-5pt]ETO.north);

\end{tikzpicture}}
\caption{The formal hierarchy of timed opacity classes.}
\label{fig:relations}
\end{figure}
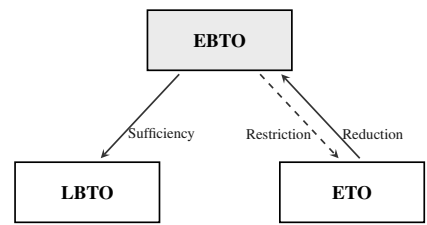

\section{Conclusion}\label{sec:conclusion}
This paper addresses the theoretical fragmentation in timed opacity by establishing a unified semantic framework. We propose a general observation model for timed automata under partial observation of locations, clocks, and events. Based on this model, we formalize the notion of EBTO. We demonstrate the expressiveness of EBTO by establishing a formal hierarchy of opacity classes. Specifically, we prove that EBTO strictly implies LBTO. It serves as a stronger verification baseline by capturing trailing time delays, an aspect that LBTO fails to observe due to suffix blindness. Furthermore, we establish a formal equivalence between EBTO and ETO under constrained observations. Future work will use this framework to investigate the decidability boundaries of timed opacity, identifying the minimal restrictions on system architecture and observer capabilities required for verification.
\section*{Statement on AI Use}
Google Gemini was utilized exclusively for language refinement and grammar correction to improve readability. All technical concepts and analyses are the authors' original work. The authors reviewed all AI-assisted text and bear full responsibility for the final manuscript.
\bibliography{library}

@article{alur1994theory,
   author = {Rajeev Alur and David L. Dill},
   title = {A theory of timed automata},
   journal = {Theor. Comput. Sci.},
   volume = {126},
   number = {2},
   pages = {183--235},
   year = {1994},
   publisher = {Elsevier}
}

@inproceedings{Marques2023,
   author = {Mariana Guimarães Marques and Raphael Julio Barcelos and João Carlos Basilio},
   title = {The use of Time-Interval Automata in the Modeling of Timed Discrete Event Systems and its Application to Opacity},
   booktitle = {IFAC-PapersOnLine},
   volume = {56},
   number = {2},
   pages = {8654--8659},
   month = {7},
   year = {2023},
   doi = {10.1016/j.ifacol.2023.10.042}
}

@article{Saboori2012,
   author = {Anooshiravan Saboori and Christoforos N. Hadjicostis},
   title = {Opacity-enforcing supervisory strategies via state estimator constructions},
   journal = {IEEE Trans. Autom. Control},
   volume = {57},
   number = {5},
   pages = {1155--1165},
   month = {5},
   year = {2011},
   doi = {10.1109/TAC.2011.2170453}
}

@article{Lin2011,
   author = {Feng Lin},
   title = {Opacity of discrete event systems and its applications},
   journal = {Automatica},
   volume = {47},
   number = {3},
   pages = {496--503},
   month = {3},
   year = {2011},
   doi = {10.1016/j.automatica.2011.01.002}
}

@article{Ammar2021,
   author = {Ikhlass Ammar and Yamen El Touati and Moez Yeddes and John Mullins},
   title = {Bounded opacity for timed systems},
   journal = {J. Inf. Secur. Appl.},
   volume = {61},
   pages = {102926},
   month = {9},
   year = {2021},
   doi = {10.1016/j.jisa.2021.102926}
}

@article{Wintenberg2022,
   author = {Andrew Wintenberg and Matthew Blischke and Stéphane Lafortune and Necmiye Ozay},
   title = {A general language-based framework for specifying and verifying notions of opacity},
   journal = {Discrete Event Dyn. Syst.},
   volume = {32},
   number = {2},
   pages = {253--289},
   month = {6},
   year = {2022},
   doi = {10.1007/s10626-021-00357-x}
}

@inproceedings{Klein2024,
   author = {Julian Klein and Paul Kogel and Sabine Glesner},
   title = {Verifying Opacity of Discrete-Timed Automata},
   booktitle = {Proc. IEEE/ACM 12th Int. Conf. Formal Methods in Software Engineering (FormaliSE)},
   pages = {55--65},
   publisher = {ACM},
   month = {4},
   year = {2024},
   doi = {10.1145/3644033.3644376}
}

@article{Rashidinejad2024,
   author = {Aida Rashidinejad and Michel Reniers and Martin Fabian},
   title = {Supervisory Control Synthesis of Timed Automata Using Forcible Events},
   journal = {IEEE Trans. Autom. Control},
   volume = {69},
   number = {2},
   pages = {1074--1080},
   month = {2},
   year = {2023},
   doi = {10.1109/TAC.2023.3275440}
}

@inproceedings{fares2013automatic,
   author = {Elie Fares and Jean-Paul Bodeveix and Mamoun Filali-Amine and Manuel Garnacho},
   title = {An Automatic Technique for Checking the Simulation of Timed Systems},
   booktitle = {Proc. 11th Int. Symp. Automated Technology for Verification and Analysis (ATVA)},
   pages = {71--86},
   publisher = {Springer},
   month = {10},
   year = {2013}
}

@article{Xie2022,
   author = {Yifan Xie and Xiang Yin and Shaoyuan Li},
   title = {Opacity Enforcing Supervisory Control Using Nondeterministic Supervisors},
   journal = {IEEE Trans. Autom. Control},
   volume = {67},
   number = {12},
   pages = {6567--6582},
   month = {12},
   year = {2021},
   doi = {10.1109/TAC.2021.3131125}
}

@inproceedings{saboori2007notions,
   author = {Anooshiravan Saboori and Christoforos N. Hadjicostis},
   title = {Notions of security and opacity in discrete event systems},
   booktitle = {Proc. 46th IEEE Conf. Decision and Control (CDC)},
   pages = {5056--5061},
   publisher = {IEEE},
   year = {2007}
}

@article{Zhang2024,
   author = {Kuize Zhang},
   title = {State-based opacity of labeled real-time automata},
   journal = {Theor. Comput. Sci.},
   volume = {987},
   pages = {114373},
   month = {3},
   year = {2024},
   doi = {10.1016/j.tcs.2023.114373}
}

@inproceedings{andre2024bright,
   author = {{\'{E}}tienne André and Sarah Dépernet and Engel Lefaucheux},
   title = {The bright side of timed opacity},
   booktitle = {Proc. Int. Conf. Formal Eng. Methods},
   pages = {51--69},
   publisher = {Springer},
   year = {2024}
}

@inproceedings{cassez2009dark,
   author = {Franck Cassez},
   title = {The dark side of timed opacity},
   booktitle = {Proc. Int. Conf. Inf. Secur. Assurance},
   pages = {21--30},
   publisher = {Springer},
   year = {2009}
}

@inproceedings{Andr2023,
   author = {{\'{E}}tienne André and Engel Lefaucheux and Didier Lime and Dylan Marinho and Jun Sun},
   title = {Configuring Timing Parameters to Ensure Execution-Time Opacity in Timed Automata},
   booktitle = {Electron. Proc. Theor. Comput. Sci. (EPTCS)},
   volume = {392},
   pages = {1--26},
   month = {10},
   year = {2023},
   doi = {10.4204/EPTCS.392.1}
}

@article{deng2025initial,
   author = {Weilin Deng and Daowen Qiu and Jingkai Yang},
   title = {Initial-Location Opacity and Infinite-Step Opacity of Timed Automata With Integer Resets},
   journal = {IEEE Control Syst. Lett.},
   year = {2025}
}

@inproceedings{An2025,
   author = {Jie An and Qiang Gao and Lingtai Wang and Naijun Zhan and Ichiro Hasuo},
   title = {The Opacity of Timed Automata},
   booktitle = {Proc. Int. Symp. on Formal Methods},
   volume = {14933},
   pages = {620--637},
   publisher = {Springer},
   year = {2024},
   doi = {10.1007/978-3-031-71162-6_32}
}

@article{HARKAT2024109891,
   author = {Houda Harkat and Luis M. Camarinha-Matos and João Goes and Hasmath F.T. Ahmed},
   title = {Cyber-physical systems security: A systematic review},
   journal = {Comput. Ind. Eng.},
   volume = {188},
   pages = {109891},
   year = {2024},
   doi = {10.1016/j.cie.2024.109891}
}

@article{Ru2009,
   author = {Yu Ru and Christoforos N. Hadjicostis},
   title = {Fault diagnosis in discrete event systems modeled by partially observed Petri nets},
   journal = {Discrete Event Dyn. Syst.},
   pages = {551--575},
   year = {2009},
   doi = {10.1007/s10626-009-0082-x}
}

@article{Shu2007,
   author = {Shaolong Shu and Feng Lin and Hao Ying},
   title = {Detectability of discrete event systems},
   journal = {IEEE Trans. Autom. Control},
   pages = {2356--2359},
   year = {2007},
   doi = {10.1109/TAC.2007.910711}
}

\appendix
\setcounter{theorem}{0}
\setcounter{proposition}{0}
\setcounter{lemma}{0} 

\renewcommand{\thetheorem}{\arabic{theorem}}
\renewcommand{\theproposition}{\arabic{proposition}}
\renewcommand{\thelemma}{\arabic{lemma}}

\renewcommand{\theHtheorem}{app.\arabic{theorem}}
\renewcommand{\theHproposition}{app.\arabic{proposition}}
\renewcommand{\theHlemma}{app.\arabic{lemma}}
\section{Proofs}
\subsection{Proof of Lemma~\ref*{lem:YM_PY}}\label{app:sec:lem:YM_PY}
\begin{lemma}\label{lem:YM_PY}
    Let $\mathcal{M}$ be the observation mapping on evolutions and $\mathcal{Y}$ be the mapping from evolutions to timed words (Definition~\ref{def:Y-mapping}). For any evolution $\rho$, the following holds:
    \[
    \mathcal{Y}\bigl(\mathcal{M}(\rho)\bigr) \;=\; P_{\Sigma_{\mathrm{obs}}}\bigl(\mathcal{Y}(\rho)\bigr).
    \]
\end{lemma}

\begin{proof}
    Let's consider an arbitrary evolution $\rho = s_0 \xrightarrow{\gamma_0} s_1 \xrightarrow{\gamma_1} \cdots \xrightarrow{\gamma_{n-1}} s_n$.

    First, we note that the global timestamps calculated by $\mathcal{Y}$ are identical for $\rho$ and $\mathcal{M}(\rho)$.
    The observation mapping $\mathcal{M}$ maps unobservable events to $\sigma_\epsilon$ but leaves time delays $\tau \in \mathbb{R}_{\geq 0}$ unchanged. Since $\mathrm{dur}(\sigma) = 0$ for any discrete event $\sigma \in \Sigma$ and naturally $\mathrm{dur}(\sigma_\epsilon) = 0$, the cumulative duration $t_k = \sum_{h=0}^{k-1} \mathrm{dur}(\gamma_h)$ at any step $k$ remains unchanged.

    It implies that we only need to compare the sequences of event-timestamp pairs. Let $K = \{ k \mid \gamma_k \in \Sigma \}$ be the set of indices where discrete events occur in $\rho$.
\begin{itemize}
    \item Right-hand side:
    The term $\mathcal{Y}(\rho)$ generates the sequence of pairs $(\gamma_k, t_k)$ for all $k \in K$.
    The projection function $P_{\Sigma_{\mathrm{obs}}}$ then retains only those pairs where $\gamma_k \in \Sigma_{\mathrm{obs}}$.
    \item Left-hand side:
    The mapping $\mathcal{M}$ converts $\gamma_k$ to $\sigma_\epsilon$ if $\gamma_k \in \Sigma_{\mathrm{uo}}$, and keeps $\gamma_k$ unchanged if $\gamma_k \in \Sigma_{\mathrm{obs}}$.
    According to Definition~\ref{def:Y-mapping}, $\mathcal{Y}$ constructs the timed word by collecting actions strictly in $\Sigma$. Since $\sigma_\epsilon \notin \Sigma$, the indices corresponding to $\Sigma_{\mathrm{uo}}$ are not included in the result. Consequently, $\mathcal{Y}(\mathcal{M}(\rho))$ consists of exactly the pairs $(\gamma_k, t_k)$ where $k \in K$ and $\gamma_k \in \Sigma_{\mathrm{obs}}$.
\end{itemize}
\end{proof}

\subsection{Proof of Lemma~\ref*{lem:obs_equiv}}\label{app:sec:lem:obs_equiv}
\begin{lemma}[\hyperlink{main:lem:obs_equiv}{Observation Correspondence for Timed Words}]\label{lem:obs_equiv}
Let $\mathcal{A}$ be a TA where locations and clocks are unobservable ($L_{\mathrm{obs}} = \varnothing$ and $C_{\mathrm{obs}} = \varnothing$).
For any two evolutions $\rho, \rho' \in \mathrm{Gen}^{\mathcal{T}}$, if their observations are equivalent in the EBTO framework, then the projections of their generated timed words are identical:
\[\begin{aligned}
&\llbracket\mathcal{M}(\rho)\rrbracket_{\equiv_{\mathrm{obs}}} = \llbracket\mathcal{M}(\rho')\rrbracket_{\equiv_{\mathrm{obs}}} \\&\implies P_{\Sigma_{\mathrm{obs}}}(\mathcal{Y}(\rho)) = P_{\Sigma_{\mathrm{obs}}}(\mathcal{Y}(\rho')).
\end{aligned}
\]
\end{lemma}

\begin{proof}
Let $\rho$ and $\rho'$ be two arbitrary evolutions.
From Lemma~\ref{lem:YM_PY}, we have the identity $P_{\Sigma_{\mathrm{obs}}}(\mathcal{Y}(\rho)) = \mathcal{Y}(\mathcal{M}(\rho))$.
Thus, the implication to be proven reduces to showing that:
\[
\llbracket\mathcal{M}(\rho)\rrbracket_{\equiv_{\mathrm{obs}}} = \llbracket\mathcal{M}(\rho')\rrbracket_{\equiv_{\mathrm{obs}}} \implies \mathcal{Y}(\mathcal{M}(\rho)) = \mathcal{Y}(\mathcal{M}(\rho')).
\]

The observational equivalence $\equiv_{\mathrm{obs}}$ is defined as the closure of silent-step equivalence ($\equiv_\epsilon$) and time additivity, continuity, and reflexivity equivalence ($\equiv_\tau$). We demonstrate that the mapping $\mathcal{Y}$ is invariant under both elementary operations applied to the observation sequence $\mathcal{M}(\rho)$.

\begin{enumerate}
    \item Invariance under $\equiv_\epsilon$:
    Suppose $\mathcal{M}(\rho')$ is derived from $\mathcal{M}(\rho)$ by inserting or removing unobservable events (mapped to $\sigma_\epsilon$).
    Recall that $\mathcal{Y}$ constructs a timed word by collecting pairs $(\sigma_k, t_k)$ where $\sigma_k \in \Sigma_{\mathrm{obs}}$.
    \begin{itemize}
        \item Event sequence: Since $\sigma_\epsilon \notin \Sigma_{\mathrm{obs}}$, the insertion or removal of $\sigma_\epsilon$ does not alter the sequence of recorded discrete events.
        \item Timestamps: The global timestamp $t_k$ is the sum of all preceding durations. Since $\mathrm{dur}(\sigma_\epsilon) = 0$, the cumulative duration remains strictly invariant.
    \end{itemize}
    Thus, $\mathcal{Y}(\mathcal{M}(\rho)) = \mathcal{Y}(\mathcal{M}(\rho'))$.

    \item  Invariance under $\equiv_\tau$:
    Suppose $\mathcal{M}(\rho')$ is derived from $\mathcal{M}(\rho)$ by merging two consecutive delays $d_1, d_2$ into $d = d_1 + d_2$ (or splitting them), or by inserting or removing zero delays ($d = 0$). Consider any observable event $\sigma_k$ occurring after these time-related modifications. Its timestamp $t_k$ is calculated as the sum of all preceding durations. By the associativity and the identity property of addition in $\mathbb{R}_{\geq 0}$:
\[\begin{aligned}
    \cdots + d_1 + d_2 + \cdots \;=\; \cdots + (d_1 + d_2) + \cdots, \\\text{and} 
    \quad \cdots + d + 0 + \cdots \;=\; \cdots + d + \cdots.
\end{aligned}
\]
    Consequently, the calculated global timestamps for all subsequent events are identical in both $\mathcal{M}(\rho)$ and $\mathcal{M}(\rho')$.
    Thus, $\mathcal{Y}(\mathcal{M}(\rho)) = \mathcal{Y}(\mathcal{M}(\rho'))$.
\end{enumerate}

Since $\mathcal{Y}$ yields identical results for any single step of $\equiv_\epsilon$ or $\equiv_\tau$, by induction, it yields identical results for any two observations in the same equivalence class.
Therefore, $\mathcal{Y}(\mathcal{M}(\rho)) = \mathcal{Y}(\mathcal{M}(\rho'))$, which implies $P_{\Sigma_{\mathrm{obs}}}(\mathcal{Y}(\rho)) = P_{\Sigma_{\mathrm{obs}}}(\mathcal{Y}(\rho'))$.
\end{proof}

\subsection{Proof of Theorem~\ref*{app:thm:lbto_to_ebto}}\label{app:sec:thm:lbto_to_ebto}
\begin{theorem}[\hyperlink{main:thm:lbto_to_ebto}{Sufficiency of EBTO for LBTO}]\label{app:thm:lbto_to_ebto}
Let $\mathcal{A}$ be a TA with SG $\mathcal{T}=\mathcal{S}(\mathcal{A})$, let $\Sigma_{\mathrm{obs}} \subseteq \Sigma$ be a given set of observable events, locations and clocks are unobservable, i.e., $L_{\mathrm{obs}} = \varnothing$ and $C_{\mathrm{obs}} = \varnothing$, and let $\mathcal{L}_s$ be a secret language. Let $(\mathcal{A}, \mathrm{Gen}_s^{\mathcal{T}}) = \mathcal{C}_{\mathrm{LE}}(\mathcal{A}, \mathcal{L}_s)$. 

The TA $\mathcal{A}$ is LBTO w.r.t. $\Sigma_{\mathrm{obs}}$ and $\mathcal{L}_s$ if it is EBTO w.r.t. the $\mathcal{M}$ and $\mathrm{Gen}_s^{\mathcal{T}}$.
\end{theorem}

\begin{proof}
Assume that the TA $\mathcal{A}$ is EBTO w.r.t. the secret evolution set $\mathrm{Gen}_s^{\mathcal{T}}$ and observation mapping $\mathcal{M}$. Let $\omega \in \mathcal{L}_{\mathrm{gen}}(\mathcal{A}) \cap \mathcal{L}_s$ be an arbitrary generated secret timed word.

To prove that $\mathcal{A}$ is LBTO, we must find a non-secret timed word $\omega' \in \mathcal{L}_{\mathrm{gen}}(\mathcal{A}) \setminus \mathcal{L}_s$ such that $P_{\Sigma_{\mathrm{obs}}}(\omega) = P_{\Sigma_{\mathrm{obs}}}(\omega')$.

Since $\omega \in \mathcal{L}_{\mathrm{gen}}(\mathcal{A})$, there exists a generated evolution $\rho \in \mathrm{Gen}^{\mathcal{T}}$ such that $\mathcal{Y}(\rho) = \omega$. Because $\omega \in \mathcal{L}_s$ and $\mathrm{Gen}_s^{\mathcal{T}} = \mathcal{Y}^{-1}(\mathcal{L}_s)$, it holds that $\rho \in \mathrm{Gen}_s^{\mathcal{T}}$. Thus, $\rho$ is a secret evolution.

By the assumption that $\mathcal{A}$ is EBTO, there exists an evolution $\rho' \in \mathrm{Gen}^{\mathcal{T}}$ such that $\rho' \notin \mathrm{Gen}_s^{\mathcal{T}}$ and $\llbracket\mathcal{M}(\rho)\rrbracket_{\equiv_{\mathrm{obs}}} = \llbracket\mathcal{M}(\rho')\rrbracket_{\equiv_{\mathrm{obs}}}$.

Let $\omega' = \mathcal{Y}(\rho')$. Since $\rho' \notin \mathrm{Gen}_s^{\mathcal{T}}$, it follows that $\omega' \notin \mathcal{L}_s$. Therefore, $\omega'$ is a non-secret timed word in $\mathcal{L}_{\mathrm{gen}}(\mathcal{A})$.

By applying Lemma~\ref{lem:obs_equiv}, the observational equivalence of the evolutions implies the equality of their projected timed words:
\[
P_{\Sigma_{\mathrm{obs}}}(\mathcal{Y}(\rho)) = P_{\Sigma_{\mathrm{obs}}}(\mathcal{Y}(\rho')).
\]
Substituting $\omega = \mathcal{Y}(\rho)$ and $\omega' = \mathcal{Y}(\rho')$ yields $P_{\Sigma_{\mathrm{obs}}}(\omega) = P_{\Sigma_{\mathrm{obs}}}(\omega')$.

For every generated secret timed word $\omega$, there exists a generated non-secret timed word $\omega'$ with an identical projection. By definition, the TA $\mathcal{A}$ is LBTO.
\end{proof}

\subsection{Proof of Proposition~\ref*{app:prop:ebto_to_lbto}}\label{app:sec:prop:ebto_to_lbto}
\begin{proposition}[\hyperlink{main:prop:ebto_to_lbto}{Strict Expressiveness of EBTO over LBTO}]\label{app:prop:ebto_to_lbto}
There exists a TA $\mathcal{A}$ and a set of secret evolutions $\mathrm{Gen}_s^{\mathcal{T}}$ such that no LBTO secret language $\mathcal{L}_s$ over any observable event set can exactly capture the same secret behaviors. Specifically, $\mathrm{Gen}_s^{\mathcal{T}} \neq \mathcal{Y}^{-1}(\mathcal{L}_s)$ for all possible $\mathcal{L}_s$.
\end{proposition}

\begin{proof}
We construct a counter-example exploiting the suffix blindness of timed words. 

Consider a TA $\mathcal{A}$, as shown in Fig.~\ref{fig:ta_structure}, with $L = \{l_0\}$, $\Sigma = \{a\}$, and a single clock $x$. The only edge is $(l_0, a, \text{true}, \{x\}, l_0)$, with $I(l_0) = \text{true}$. Let $\mathcal{T} = \mathcal{S}(\mathcal{A})$ be its semantic graph. A fragment of $\mathcal{T}$ is illustrated in Fig.~\ref{fig:sg_fragment}.
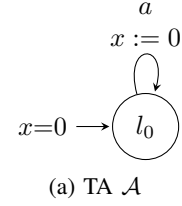
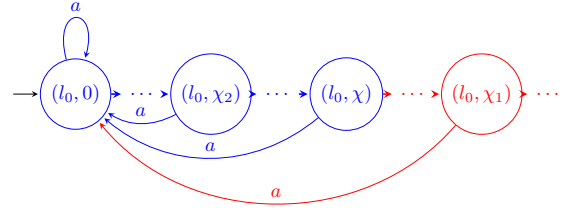
\begin{figure}[t]
    \centering
    \begin{subfigure}{\columnwidth}
        \centering
        \scalebox{1.0}{
        \begin{tikzpicture}[->, >=stealth, shorten >=1pt, node distance=2cm, on grid, auto]
            \node[state, initial, initial text={$x{=}0$}] (s0) {$l_0$};
            \path (s0) edge [loop above] node [align=center] {$a$\\$x:=0$} (s0);
        \end{tikzpicture}}
        \caption{TA $\mathcal{A}$}
        \label{fig:ta_structure}
    \end{subfigure}

    \vspace{1.5em} 

    \begin{subfigure}{\columnwidth}
  \centering
  \scalebox{0.75}{
  \begin{tikzpicture}[->, >=stealth, shorten >=1pt, node distance=2cm, on grid, auto]
    \node[state, initial, initial text={}, draw=blue, text=blue] (s0) {$(l_0, 0)$};
    \node[right=1.2cm of s0, text=blue] (d1) {$\dots$};
    \node[state, draw=blue, text=blue, right=1.2cm of d1] (s2) {$(l_0, \chi_2)$};
    \node[right=1.2cm of s2, text=blue] (d_pre_chi) {$\dots$};
    
    \node[state, draw=blue, text=blue, right=1.2cm of d_pre_chi] (s_chi) {$(l_0, \chi)$};
    
    \node[right=1.2cm of s_chi, text=red] (d2) {$\dots$};
    \node[state, draw=red, text=red, right=1.2cm of d2] (s1) {$(l_0, \chi_1)$};
    \node[right=1.2cm of s1, text=red] (d3) {$\dots$};
    \path [blue, thick] (s0) edge (d1)
                        (d1) edge (s2)
                        (s2) edge (d_pre_chi)
                        (d_pre_chi) edge (s_chi);
    \path [red, thick]  (s_chi) edge (d2)
                        (d2) edge (s1)
                        (s1) edge (d3);

    \path (s0) edge [loop above, blue] node {$a$} (s0)
          (s2) edge [bend left=30, blue] node [above] {$a$} (s0)
          (s_chi) edge [bend left=40, blue] node [above] {$a$} (s0)
          (s1) edge [bend left=50, red] node [above] {$a$} (s0);

  \end{tikzpicture}}
        \caption{Fragment of semantic graph $\mathcal{T}$}
        \label{fig:sg_fragment}
    \end{subfigure}

    \caption{The counter-example for suffix blindness: (a) shows the TA structure, and (b) illustrates states in the semantic graph. Blue states represent non-secret behaviors ($\tau \le \chi$), while red states represent secret behaviors ($\tau > \chi$).}
    \label{fig:combined_suffix_blindness}
\end{figure}

We define the EBTO secret evolutions $\mathrm{Gen}_s^{\mathcal{T}}$ as the set of evolutions ending with event $a$ followed by a trailing delay greater than a constant $\chi \in \mathbb{R}_{>0}$:
$$
\mathrm{Gen}_s^{\mathcal{T}} = \bigcup \left\{ \llbracket\rho\rrbracket_{\equiv_\tau} \;\middle|\; \rho = \rho_{\mathrm{pre}} \xrightarrow{a} (l_0, 0) \xrightarrow{\tau} (l_0, \tau) \land \tau > \chi \right\}.
$$

Consider two evolutions $\rho_1$ and $\rho_2$ generated by $\mathcal{T}$:
\begin{itemize}
    \item $\rho_1 \;=\; (l_0, 0) \xrightarrow{a} (l_0, 0) \xrightarrow{\chi_1} (l_0, \chi_1)$ with $\chi_1 > \chi$. By definition, $\rho_1 \in \mathrm{Gen}_s^{\mathcal{T}}$.
    \item $\rho_2 \;=\; (l_0, 0) \xrightarrow{a} (l_0, 0) \xrightarrow{\chi_2} (l_0, \chi_2)$ with $\chi_2 \leq \chi$. By definition, $\rho_2 \notin \mathrm{Gen}_s^{\mathcal{T}}$.
\end{itemize}

Applying the mapping $\mathcal{Y}$, both evolutions generate the identical timed word because the trailing delays produce no discrete events:
$$
\mathcal{Y}(\rho_1) = (a, 0) = \mathcal{Y}(\rho_2).
$$
Let $\omega = (a, 0)$. Assume for contradiction that there exists a secret language $\mathcal{L}_s$ such that $\mathrm{Gen}_s^{\mathcal{T}} = \mathcal{Y}^{-1}(\mathcal{L}_s)$. 
We evaluate $\omega$ against $\mathcal{L}_s$:
\begin{itemize}
    \item If $\omega \in \mathcal{L}_s$, then $\mathcal{Y}^{-1}(\mathcal{L}_s)$ must contain all evolutions mapping to $\omega$. Thus, $\rho_2 \in \mathcal{Y}^{-1}(\mathcal{L}_s)$. However, $\rho_2 \notin \mathrm{Gen}_s^{\mathcal{T}}$, causing a contradiction.
    \item If $\omega \notin \mathcal{L}_s$, then $\mathcal{Y}^{-1}(\mathcal{L}_s)$ cannot contain any evolution mapping to $\omega$. Thus, $\rho_1 \notin \mathcal{Y}^{-1}(\mathcal{L}_s)$. However, $\rho_1 \in \mathrm{Gen}_s^{\mathcal{T}}$, causing a contradiction.
\end{itemize}
Therefore, no such $\mathcal{L}_s$ exists. EBTO can specify duration-based suffix secrets that are structurally impossible to represent in LBTO.
\end{proof}

\subsection{Proof of Lemma~\ref*{lem:duration_strict}}\label{app:sec:lem:duration_strict}
\begin{lemma}[\hyperlink{main:lem:duration_strict}{Observation Correspondence for final Reachability}]\label{lem:duration_strict}
Let $\mathcal{A}$ be a TA with a final location $l_f$. Under the observation model where only the final location is observable ($L_{\mathrm{obs}} = \{l_f\}$) and all events and clocks are unobservable ($\Sigma_{\mathrm{obs}} = \varnothing$, $C_{\mathrm{obs}} = \varnothing$), for any two evolutions $\rho, \rho' \in \mathrm{Gen}_{l_f}^{\mathcal{T}}$, they are observably equivalent if and only if their total durations are equal:
\[\begin{aligned}
\llbracket \mathcal{M}(\rho) \rrbracket_{\equiv_{\mathrm{obs}}} = \llbracket \mathcal{M}(\rho') \rrbracket_{\equiv_{\mathrm{obs}}} \iff \\\mathsf{Duration}(\rho) = \mathsf{Duration}(\rho').
\end{aligned}
\]
\end{lemma}

\begin{proof}
By definition, $\rho, \rho' \in \mathrm{Gen}_{l_f}^{\mathcal{T}}$ spend their entire execution in unobservable locations before terminating exactly upon the first entry into $l_f$. Their mapped observations $\mathcal{M}(\rho)$ and $\mathcal{M}(\rho')$ consist of an initial unobservable state $(l_\epsilon, u_\epsilon)$, interspersed with unobservable events $\sigma_\epsilon$ and time delays, terminating with a transition into $(l_f, u_\epsilon)$.

$(\implies)$ Assume $\llbracket \mathcal{M}(\rho) \rrbracket_{\equiv_{\mathrm{obs}}} = \llbracket \mathcal{M}(\rho') \rrbracket_{\equiv_{\mathrm{obs}}}$. The equivalence relation $\equiv_{\mathrm{obs}}$ applies $\equiv_\epsilon$ (which adds or removes $\sigma_\epsilon$ with $0$ duration) and $\equiv_\tau$ (which splits or merges delays). These operations preserve the cumulative time elapsed. Since both observations end with the identical observable state change to $l_f$, the total time accumulated before this observable change is equal. Thus, $\mathsf{Duration}(\rho) = \mathsf{Duration}(\rho')$.

$(\impliedby)$ Assume $\mathsf{Duration}(\rho) = \mathsf{Duration}(\rho') = d$. For both $\mathcal{M}(\rho)$ and $\mathcal{M}(\rho')$, all intermediate discrete transitions occur between unobservable locations, yielding $(l_\epsilon, u_\epsilon) \xrightarrow{\sigma_\epsilon} (l_\epsilon, u_\epsilon)$. By applying $\equiv_\epsilon$, we remove all these intermediate $\sigma_\epsilon$ symbols. By applying $\equiv_\tau$, we merge all fragmented delays. However, the final transition enters $l_f$, meaning the observable state changes from $(l_\epsilon, u_\epsilon)$ to $(l_f, u_\epsilon)$. Consequently, the equivalence $\equiv_\epsilon$ cannot be applied to erase this specific transition. Both mapped observations are strictly reduced to the identical canonical sequence: 
$$
(l_\epsilon, u_\epsilon) \xrightarrow{d} (l_\epsilon, u_\epsilon) \xrightarrow{\sigma_\epsilon} (l_f, u_\epsilon).
$$
Thus, $\llbracket \mathcal{M}(\rho) \rrbracket_{\equiv_{\mathrm{obs}}} = \llbracket \mathcal{M}(\rho') \rrbracket_{\equiv_{\mathrm{obs}}}$.
\end{proof}

\subsection{Proof of Theorem~\ref*{thm:eto_to_ebto}}\label{app:sec:thm:eto_to_ebto}
\begin{theorem}[\hyperlink{main:thm:eto_to_ebto}{Reduction of ETO to EBTO}]\label{thm:eto_to_ebto}
Let $\mathcal{A}$ be a TA with private location $l_{\mathrm{priv}}$ and final location $l_f$. Let $L_{\mathrm{obs}}=\{l_f\}$, $\Sigma_{\mathrm{obs}} = \varnothing$, and $C_{\mathrm{obs}} = \varnothing$. Let $(\mathcal{A}, \mathrm{Gen}_s^{\mathcal{T}}) = \mathcal{C}_{\mathrm{ExE}}(\mathcal{A},l_{\mathrm{priv}},l_f)$ be the converted system. 

The TA $\mathcal{A}$ is ETO w.r.t. $l_{\mathrm{priv}},l_f$ iff it is EBTO w.r.t. the observation mapping $\mathcal{M}$ and secret evolutions $\mathrm{Gen}_s^{\mathcal{T}}$.
\end{theorem}

\begin{proof}
$(\implies)$ Assume $\mathcal{A}$ is ETO. Let $\rho \in \mathrm{Gen}_s^{\mathcal{T}}$ be an arbitrary secret evolution. By definition, $\rho \in \mathrm{Gen}_{l_f}^{\mathcal{T}}$ and visits $l_{\mathrm{priv}}$. Let its duration be $d = \mathsf{Duration}(\rho)$.

By the definition of ETO, the existence of this private run implies the existence of a public run in $\overline{\mathsf{Visit}}^{\mathrm{priv}}(\mathcal{A})$ with the identical duration $d$. Let $\rho'$ be the evolution generating this public run. Since $\rho'$ terminates upon entering $l_f$, $\rho' \in \mathrm{Gen}_{l_f}^{\mathcal{T}}$. Since it does not visit $l_{\mathrm{priv}}$, $\rho' \notin \mathrm{Gen}_s^{\mathcal{T}}$. 

For $\rho$ and $\rho'$, we have established $\rho, \rho' \in \mathrm{Gen}_{l_f}^{\mathcal{T}}$ and $\mathsf{Duration}(\rho) = \mathsf{Duration}(\rho')$. By Lemma~\ref{lem:duration_strict}, this implies $\llbracket \mathcal{M}(\rho) \rrbracket_{\equiv_{\mathrm{obs}}} = \llbracket \mathcal{M}(\rho') \rrbracket_{\equiv_{\mathrm{obs}}}$. The system $\mathcal{A}$ is EBTO.

$(\impliedby)$ Assume $\mathcal{A}$ is EBTO. Let $\mathsf{Run}_{\mathrm{priv}} \in \mathsf{Visit}^{\mathrm{priv}}(\mathcal{A})$ be an arbitrary private run with duration $d$. By the definition of $\mathcal{C}_{\mathrm{ExE}}$, there exists a corresponding secret evolution $\rho \in \mathrm{Gen}_s^{\mathcal{T}}$ terminating at $l_f$ with $\mathsf{Duration}(\rho) = d$.

Since $\mathcal{A}$ is EBTO, there exists an evolution $\rho' \in \mathrm{Gen}^{\mathcal{T}}$ such that $\rho' \notin \mathrm{Gen}_s^{\mathcal{T}}$ and $\llbracket \mathcal{M}(\rho) \rrbracket_{\equiv_{\mathrm{obs}}} = \llbracket \mathcal{M}(\rho') \rrbracket_{\equiv_{\mathrm{obs}}}$. 

Since $\rho \in \mathrm{Gen}_{l_f}^{\mathcal{T}}$, its mapped observation $\mathcal{M}(\rho)$ records a single observable location change into $l_f$ exactly at its termination time $d$. Because observational equivalence ($\equiv_{\mathrm{obs}}$) strictly preserves the occurrence and timing of observable state changes, $\mathcal{M}(\rho')$ must also record its first and only entry into $l_f$ exactly at time $d$, with no subsequent delays. This guarantees that $\rho'$ terminates upon entering $l_f$, yielding $\rho' \in \mathrm{Gen}_{l_f}^{\mathcal{T}}$ and $\mathsf{Duration}(\rho') = d$. 

Since $\rho' \in \mathrm{Gen}_{l_f}^{\mathcal{T}}\setminus\mathrm{Gen}_s^{\mathcal{T}}$, it implies that $\rho'$ does not visit $l_{\mathrm{priv}}$. The run $\Psi(\rho')$ extracted from $\rho'$ is therefore a valid public run in $\overline{\mathsf{Visit}}^{\mathrm{priv}}(\mathcal{A})$ with duration $d$. Thus, $\mathsf{DVisit}^{\mathrm{priv}}(\mathcal{A}) \subseteq \mathsf{D\overline{Visit}^{\mathrm{priv}}}(\mathcal{A})$. The system $\mathcal{A}$ is ETO.
\end{proof}

\subsection{Proof of Proposition~\ref*{prop:ebto_to_eto}}\label{app:sec:prop:ebto_to_eto}
\begin{proposition}[\hyperlink{main:prop:ebto_to_eto}{Strict Expressiveness of EBTO over ETO}]\label{prop:ebto_to_eto}
There exists an EBTO instance $(\mathcal{A}, \mathrm{Gen}_s^{\mathcal{T}})$ such that for any TA $\mathcal{A}'$ and $l_{\mathrm{priv}}$ and $l_f$, the ETO property of $(\mathcal{A}', l_{\mathrm{priv}}$ and $l_f)$ is not equivalent to the EBTO property of $(\mathcal{A}, \mathrm{Gen}_s^{\mathcal{T}})$.
\end{proposition}

\begin{proof}
We construct a counter-example where the secrecy depends on the non-commutative temporal ordering of delays, which is erased by the summation in ETO.

Consider a TA, as shown in Fig.~\ref{fig:timed_automaton}, $\mathcal{A} = (L, l_0, C, \Sigma, E, I)$ and defined as follows:
\begin{itemize}
    \item $L = \{l_0, l_1, l_2, l_f\}$; $C = \{x\}$; $\Sigma = \{a, b\}$; $I(l) = \text{true}$ for all $l$.
    \item Let $c_1, c_2 \in \mathbb{R}_{>0}$ be two constants such that $c_1 \neq c_2$.
    \item The set of edges $E$ forms two paths:
    \begin{itemize}
        \item \textit{Path 1:} $l_0 \xrightarrow{a, x=c_1, \{x\}} l_1$ and $l_1 \xrightarrow{b, x=c_2, \varnothing} l_f$.
        \item \textit{Path 2:} $l_0 \xrightarrow{a, x=c_2, \{x\}} l_2$ and $l_2 \xrightarrow{b, x=c_1, \varnothing} l_f$.
    \end{itemize}
\end{itemize}

\begin{figure}[t]
    \centering
    \begin{tikzpicture}[
        ->,                      
        >=Stealth,               
        shorten >=1pt,           
        auto,                    
        node distance=1.5cm,       
        semithick,               
        state/.style={circle, draw, minimum size=1cm} 
    ]

        \node[state, initial, initial text={}] (l0) {$l_0$};
        
        \node[state] (l1) [above right=of l0] {$l_1$};
        \node[state] (l2) [below right=of l0] {$l_2$};
        
        \node[state, accepting] (lf) [below right=of l1] {$l_f$};


        
        \path (l0) edge node[align=center] {$a$\\ $x=c_1$\\ $x:=0$} (l1)
              (l1) edge node[align=center] {$b$\\ $x=c_2$\\} (lf);

        \path (l0) edge [swap] node[align=center] {$a$\\ $x=c_2$\\ $x:=0$} (l2)
              (l2) edge [swap] node[align=center] {$b$\\ $x=c_1$} (lf);

    \end{tikzpicture}
    \caption{Timed Automaton $\mathcal{A}$ illustrating path dependency on clock constraints.}
    \label{fig:timed_automaton}
\end{figure}
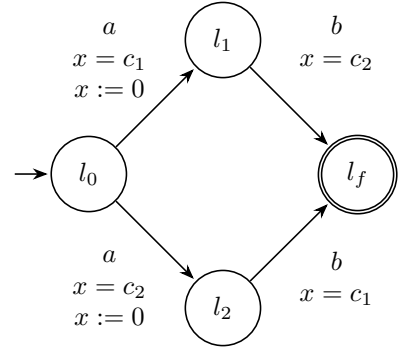

The TA generates two strict evolutions:
\begin{itemize}
    \item $\rho_1$ (via $l_1$):
    \[
    \rho_1 = (l_0, 0) \xrightarrow{c_1} (l_0, c_1) \xrightarrow{a} (l_1, 0) \xrightarrow{c_2} (l_1, c_2) \xrightarrow{b} (l_f, c_2).
    \]
    \item $\rho_2$ (via $l_2$):
    \[
    \rho_2 = (l_0, 0) \xrightarrow{c_2} (l_0, c_2) \xrightarrow{a} (l_2, 0) \xrightarrow{c_1} (l_2, c_1) \xrightarrow{b} (l_f, c_1).
    \]
\end{itemize}

Let $\mathrm{Gen}_s^{\mathcal{T}} = \{\rho_1\}$, $\Sigma_{\mathrm{obs}} = \{a, b\}, L_{\mathrm{obs}}=\varnothing$, and $ C_{\mathrm{obs}}=\varnothing$. Since $c_1 \neq c_2$, then $\llbracket \mathcal{M}(\rho_1) \rrbracket_{\equiv_{\mathrm{obs}}} \neq \llbracket \mathcal{M}(\rho_2) \rrbracket_{\equiv_{\mathrm{obs}}}$. Thus, the TA is not EBTO.

To attempt capturing $\rho_1$ as a secret in ETO, we set $l_{\mathrm{priv}} = l_1$ and final location as $l_f$.
\begin{itemize}
    \item The duration of the private run $\Psi(\rho_1)$ is $c_1 + c_2$.
    \item The duration of the public run $\Psi(\rho_2)$ is $c_2 + c_1$.
\end{itemize}
Due to the commutativity of addition in $\mathbb{R}_{\geq 0}$, $c_1 + c_2 = c_2 + c_1$. Thus, $\mathsf{DVisit}^{\mathrm{priv}}(\mathcal{A}) \subseteq \mathsf{D\overline{Visit}^{\mathrm{priv}}}(\mathcal{A})$, and the TA is ETO.

Even with strict evolutions and ideal location labeling, ETO is structurally blind to the permutation of time intervals, whereas EBTO correctly identifies the leakage.
\end{proof}

\end{document}